\documentclass[]{interact}

\usepackage{epstopdf}
\usepackage[caption=false]{subfig}
\usepackage[longnamesfirst,sort]{natbib}
\bibpunct[, ]{(}{)}{;}{a}{,}{,}

\usepackage{amsmath,amssymb,amsfonts,mathtools}
\usepackage{graphicx}
\usepackage{float}
\usepackage{enumitem}
\usepackage{algorithm}
\usepackage{algpseudocode}
\usepackage{tabularx}
\usepackage{ragged2e}
\usepackage{microtype}
\usepackage{url}
\usepackage[hidelinks]{hyperref}
\graphicspath{{./}}

\theoremstyle{plain}
\newtheorem{theorem}{Theorem}[section]
\newtheorem{lemma}[theorem]{Lemma}
\newtheorem{corollary}[theorem]{Corollary}

\theoremstyle{definition}

\theoremstyle{remark}
\newtheorem{remark}[theorem]{Remark}

\begin{document}

\articletype{ARTICLE}

\title{Climate-Conditioned Cascade Modeling for Multi-Peril Reinsurance:\\
Analysis and Controlled Numerical Applications}

\author{
\name{N.~Karimi\textsuperscript{1}\thanks{CONTACT N.~Karimi. Email: kariminader22@gmail.com}, E.~Salavati\textsuperscript{a} and F.~Shokrollahi\textsuperscript{b}}
\affil{\textsuperscript{1,a}Department of Applied Mathematics, Faculty of Mathematics and Computer Science, Amirkabir University of Technology, No. 424, Hafez Ave., 15914, Tehran, Iran;\\ \textsuperscript{b}Department of Mathematics and Statistics, University of Vaasa, Vaasa, Finland}
}

\maketitle

\begin{abstract}
Climate perils are linked through event ordering and state-dependent propagation, features not fully captured by joint loss distributions alone. This paper develops a Cascading Climate Risk Network (CCRN) for multi-peril reinsurance that separates calendar-scale climate conditioning from within-event propagation on a directed acyclic graph (DAG). The model combines complementary-log-log triggering hazards with bounded severity activation, mapping physical states to insured losses via a capacity-bounded demand-surge transformation.

For fixed shocks, the event-scale cascade reaches a unique finite-step closure. Monotone comparative statics provide a pathwise upper-corner loss bound over rectangular stress sets, yielding a transparent contract-level stress-testing guarantee under common aleatory inputs.

Comprehensive numerical experiments, including copula and Bayesian-network benchmarks, sensitivity analyses, and uncertainty propagation, demonstrate that while central layer prices remain robust across matched-marginal dependence structures, far-tail and high-layer behaviors differ materially. Directional propagation, annual event frequency, and dependence strength emerge as the principal risk drivers. The study provides a controlled synthetic verification of the proposed architecture.
\end{abstract}


\section{Introduction}
\label{intro}

Dependence modeling is central to multi-peril insurance and reinsurance because aggregate losses depend on both marginal hazard behaviour and cross-peril association. Copulas provide a flexible separation of marginal distributions and dependence and can represent nonlinear, asymmetric, and time-varying association when appropriately specified \citep{frees1998,joe1997,embrechts2002,aas2009,czado2010,brechmann2013,hafner2012,creal2013}. They therefore remain important actuarial benchmarks. Related financial-network
research shows how directed connections can amplify shocks at the system level
\citep{allen2000,acemoglu2015,battiston2012,haldane2011,cummins2014,eling2019};
the present paper translates that broad amplification idea into a physical
multi-peril and reinsurance setting rather than importing a balance-sheet
contagion model directly.

A different modeling requirement arises when the scientific question concerns an explicitly ordered sequence of states. Antecedent fuel aridity may alter wildfire susceptibility; a realized wildfire may then modify vegetation cover and slope conditions; and subsequent intense rainfall may activate post-fire debris-flow risk. Reviews of interacting and cascading hazards emphasize that triggering, amplification, and temporal ordering should be distinguished from coincident or statistically dependent hazards \citep{gill2014,tilloy2019,zscheischler2020,raymond2020}. In such settings, a graph can state which transitions are mechanistically admissible, while a statistical model estimates their magnitudes.

The proposed Cascading Climate Risk Network (CCRN) is not presented as a new theory of directed graphical models. Directed acyclic graphs, event trees, multi-state hazards, and self- or mutually exciting point processes already provide mature representations of conditional structure and temporal propagation \citep{lauritzen1996,hawkes1971}. Instead, the CCRN is developed as an actuarial architecture that links a two-time-scale climate/event representation, climate-conditioned triggering hazards, separate occurrence and conditional-severity states, a capacity-bounded financial loss map, and an explicit annual-aggregate reinsurance contract \citep{karimi2026cat}. This positioning keeps the methodological claim aligned with the paper's analytical and numerical evidence. Table~\ref{tab:adjacent-models} summarizes the intended division of labour between the CCRN and adjacent model classes.

\begin{table}[H]
\centering
\small
\caption{Relationship of the CCRN to adjacent dependence and propagation models. The entries describe typical formulations rather than universal limitations of each model class.}
\label{tab:adjacent-models}
\begin{tabularx}{\textwidth}{@{}>{\RaggedRight\arraybackslash}p{0.18\textwidth}>{\RaggedRight\arraybackslash}X>{\RaggedRight\arraybackslash}X@{}}
\toprule
\textbf{Model class} & \textbf{Main capability} & \textbf{Relation to the CCRN} \\
\midrule
Copula model & Flexible nonlinear and tail dependence for a joint outcome distribution. & Benchmark, or a residual-dependence layer when event ordering is not the main object. \\
Bayesian/DAG model & Conditional factorization and transparent admissible pathways. & Closest structural relative; the CCRN adds onset hazards, conditional severity, bounded loss, and contract aggregation. \\
Hawkes/point process & Continuous-time recurrent excitation, including feedback. & More natural when exact event times and recurrent feedback are central; the present CCRN assumes an event-scale DAG. \\
CCRN & Ordered occurrence-severity states translated into insured loss. & Scenario-analysis framework studied here. \\
\bottomrule
\end{tabularx}
\end{table}

Table~\ref{tab:adjacent-models} positions the proposed framework relative to three adjacent model families. Copulas remain flexible models for the joint distribution of final outcomes, Bayesian or DAG formulations are the closest structural relatives because they encode conditional pathways, and Hawkes-type models are more natural when recurrent excitation and feedback in continuous time are essential. Within this positioning, the paper makes four connected contributions. First, it couples an ordered occurrence--severity cascade to the capacity-bounded loss transformation in Equation~\eqref{eq:bounded-loss} and the annual-aggregate contract map in Equation~\eqref{eq:annual-contract-loss}. Second, it establishes pathwise, rather than only in-expectation, comparative statics for ground-up and ceded contract losses under common aleatory inputs through Corollary~\ref{cor:comparative-statics} and Theorem~\ref{thm:certified-bound}. Third, it separates calendar-scale climate conditioning from within-event physical propagation so that the two sources of variation are not conflated. Fourth, it provides a matched-marginal benchmark protocol for comparing induced joint distributions while explicitly limiting the interpretation to an in-sample dependence-shape comparison rather than an out-of-sample superiority test. The novelty therefore lies in this actuarial composition and its contract-level monotonicity guarantee, not in a claim to have invented DAG propagation itself.

\begin{figure}[H]
\centering
\includegraphics[width=0.96\textwidth]{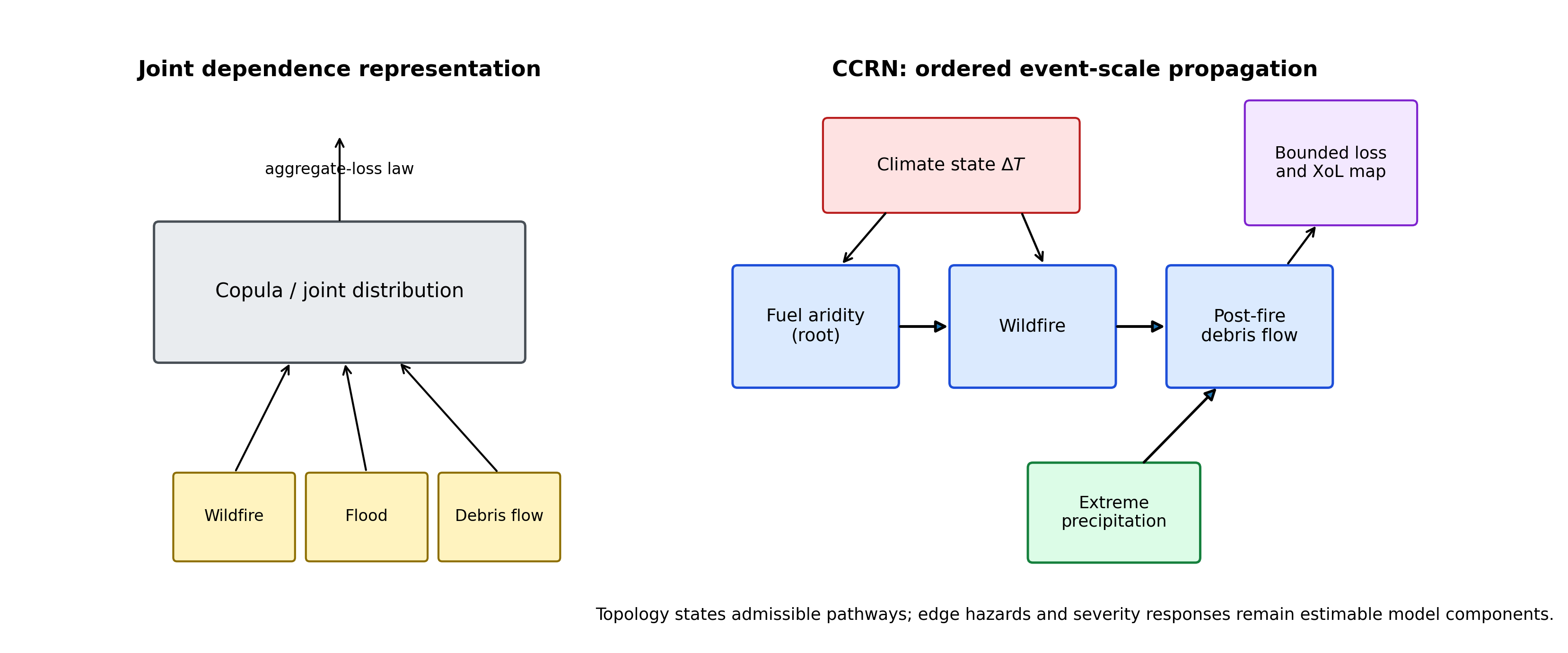}
\caption{Computational architecture of the CCRN. Exogenous event-window shocks and climate-conditioned edge hazards generate the cascade closure; the bounded loss map and annual frequency model then determine the contract loss.}
\label{fig:ccrn-architecture}
\end{figure}

Figure~\ref{fig:ccrn-architecture} summarizes the full computational flow. The left panel represents the benchmark view in which peril outcomes enter a joint dependence model and are mapped directly to an aggregate-loss law. The right panel makes the event ordering explicit: the climate state modifies the fuel-aridity-to-wildfire and wildfire-rainfall-to-debris-flow channels, the directed graph determines which transitions are admissible, and the final physical severities are passed through the bounded loss and excess-of-loss mappings. The diagram is included to show that climate conditioning, event-scale propagation, and financial aggregation are separate layers of the model rather than interchangeable sources of dependence.

The mathematical analysis establishes finite-step closure and uniqueness on a finite event-scale DAG. Monotonicity under common random numbers then yields an exact upper-corner bound for a documented rectangular stress set. The bound is deliberately narrow: it does not convert a correlated confidence region into a rectangle, identify causal effects, or replace a full parameter- and climate-uncertainty analysis.

The numerical evaluation is organized as a controlled synthetic verification study. It combines an annual frequency--severity construction with matched-marginal copula benchmarks, Student-$t$ tail sensitivity, an exogenous equicorrelation grid, attachment--limit contract tests, climate-conditioned frequency, residual latent dependence, structural ablation, structured one-at-a-time parameter sensitivity, synthetic parameter recovery, uncertainty propagation, a static Bayesian-network benchmark, and independent Monte Carlo replications.
\subsection{Economic Motivation and the Scope of Statistical Dependence Models}\label{Eco}

Portfolio diversification remains fundamental to reinsurance, but its effectiveness depends on the joint loss law, concentration of exposure, contract wording, and the stability of relationships across climate states. A static dependence model may be adequate for some pricing tasks, whereas an explicitly ordered model is useful when the event definition itself contains triggering or state inheritance. The distinction is therefore one of modeling purpose rather than a categorical replacement of copulas.

Consider the sequence fuel aridity $\rightarrow$ wildfire $\rightarrow$ post-fire debris-flow susceptibility, with intense precipitation acting as a second parent of debris flow. The first transition concerns wildfire susceptibility, while the second requires both a realized burn state and a storm-related driver. A joint distribution can reproduce the observed association among final losses, but it does not by itself state which pathways are physically admissible or how the conditional transition changes under a climate scenario. The CCRN makes those assumptions explicit and therefore exposes them to scientific review.

The financial transformation is also separated from the physical cascade. Regional scarcity of labour, materials, and temporary accommodation can amplify repair costs after clustered disasters \citep{hallegatte2008,botzen2021}. Because unconstrained multipliers can create unrealistic losses, the model applies demand surge on the odds scale and caps each component loss at a declared maximum modeled loss. This is a scenario-analysis device; empirical use requires exposure, vulnerability, repair-cost, and contract data estimated within a frozen training sample.

Accordingly, the paper does not claim that a cascade specification always produces a heavier tail than a copula or that climate change mechanically invalidates diversification. The controlled experiments ask narrower questions: whether the implementation respects the mathematical bounds, how directional propagation changes the loss distribution under fixed assumptions, and how sensitive an illustrative annual aggregate layer is to documented modeling components.

The paper is organized into six top-level sections. Section~\ref{math}
defines the CCRN, its financial mapping, and its analytical properties.
Section~\ref{data} gives the core estimation and data protocol, with the detailed
implementation workflow in Appendix~\ref{app:data-workflow}.
Section~\ref{num} presents the numerical study, dependence and contract
sensitivity, uncertainty propagation, and independent replication analysis.
Section~\ref{analy} discusses computational complexity and numerical verification, and
Section~\ref{con} concludes.
\section{CCRN Model and Analytical Properties}
\label{math}

We represent the insured physical system by a fixed directed graph
$\mathcal{G}=(\mathcal{V},\mathcal{E}_{\mathrm{phys}})$, where
$\mathcal{V}=\{v_1,\ldots,v_n\}$ is a set of peril-exposure components
and $\mathcal{E}_{\mathrm{phys}}\subseteq\mathcal{V}\times\mathcal{V}$
contains mechanistically admissible propagation directions. The graph
topology is fixed over the modeled event horizon, whereas the numerical
strength of its propagation channels may vary across calendar periods,
spatial units, and climate scenarios.

We distinguish three indices. The index $y$ denotes the calendar period
or climate-scenario evaluation time, $k\in\{1,\ldots,K\}$ indexes spatial
units, and $r\in\mathbb{N}_0$ denotes the within-event propagation step.
For a fixed $y$ and $k$, climate and exposure conditions are held constant
during the comparatively short event-scale cascade. This distinction
prevents long-term climate evolution from being conflated with short-term
physical propagation.

In the numerical study, $y$ indexes the simulated contract year and the
calendar-scale climate input assigned to that year; the climate state is fixed
within all event windows of that year. This convention avoids using separate
symbols for two indices that coincide in the controlled design.

\begin{table}[H]
\centering
\small
\caption{Core notation used in the cascade and contract mapping.}
\label{tab:notation}
\begin{tabularx}{\textwidth}{@{}p{0.18\textwidth}>{\RaggedRight\arraybackslash}X p{0.20\textwidth}@{}}
\toprule
\textbf{Symbol} & \textbf{Meaning} & \textbf{First definition} \\
\midrule
$w_{ij,k}(y)$, $H_{ij,k}(y)$ & Edge triggering probability and its cumulative hazard & Eq.~\eqref{eq:dynamic-edge-weight} \\
$Z_j^{(r)}$, $S_j^{(r)}$ & Occurrence indicator and normalized conditional severity & Section~\ref{subsec:cascade-construction} \\
$\pi_{j,k}^{(r+1)}$ & Total onset probability of a non-root node & Eq.~\eqref{eq:onset-probability} \\
$X_{j,k}^{(r)}$, $\Psi_j$ & Physical stress index and bounded severity response & Eqs.~\eqref{eq:physical-stress}--\eqref{eq:severity-response} \\
$L_j$, $L_{\mathrm{ann},y}$ & Component ground-up loss and annual aggregate ground-up loss & Eqs.~\eqref{eq:bounded-loss}, \eqref{eq:annual-contract-loss} \\
$D$, $M$, $L_{R,y}$ & Annual aggregate deductible, limit, and ceded loss & Eq.~\eqref{eq:annual-contract-loss} \\
\bottomrule
\end{tabularx}
\end{table}

Table~\ref{tab:notation} consolidates the principal symbols used in the
propagation and financial layers of the model. It distinguishes edge
probabilities and cumulative hazards from realized occurrence and severity
states, and it links the physical stress variables, ground-up losses, and
contract quantities to the equations in which they are first defined. The
table therefore serves as a reference for the notation used in the subsequent
construction and proofs.

\subsection{Climate-Conditioned Cascade Construction}
\label{subsec:cascade-construction}

The physical cascade is constructed in four linked stages: climate-conditioned edge hazards, occurrence states, conditional severity, and a realized propagation recursion. We begin with the edge hazards. Let $\mathcal{L}_{ij}$ denote a prespecified physically plausible lag
window for the directed edge $(i,j)$. Let $R_{j,k,y}=1$ indicate that
component $j$ in spatial unit $k$ is at risk of a new onset, and let
$\mathcal{H}_{k,y-1}$ contain the information available before that
onset.

The edge quantity $w_{ij,k}(y)\in(0,1)$ is defined as the incremental
triggering probability associated with a fully active parent $v_i$ over
the lag window $\mathcal{L}_{ij}$, conditional on the declared risk set
and covariates. It is an edge-specific probability contribution and
must not be interpreted as the total probability of onset of node $v_j$.

We specify the edge probability using the complementary log-log relation
$\log[-\log(1-w_{ij,k}(y))]
=\alpha_{ij}+u_{ij,k}+\beta_{ij}z_k(y)
+\boldsymbol{\gamma}_{ij}^{\top}\mathbf{c}_k(y)$, where $z_k(y)$ is a
climate anomaly relative to a declared reference period,
$\mathbf{c}_k(y)$ contains prespecified time-varying covariates or
confounders, and $u_{ij,k}$ is a spatial random effect.

At the reference climate state, the baseline edge probability is
$p_{ij,k}^{0}=1-\exp\{-\exp(\alpha_{ij}+u_{ij,k})\}$. The corresponding
climate-conditioned probability can therefore be written as
\begin{equation}
w_{ij,k}(y)
=
1-
\left(1-p_{ij,k}^{0}\right)^{
\exp\left[
\beta_{ij}z_k(y)
+
\boldsymbol{\gamma}_{ij}^{\top}\mathbf{c}_k(y)
\right]
}.
\label{eq:dynamic-edge-weight}
\end{equation}

Equation~\eqref{eq:dynamic-edge-weight} guarantees
$0<w_{ij,k}(y)<1$ for every finite covariate value. A positive
$\beta_{ij}$ implies that the edge probability increases with the
specified climate anomaly, but its sign is constrained only when a
mechanistic argument is declared before estimation. Otherwise,
$\beta_{ij}$ is estimated without a sign restriction and reported with
an uncertainty interval.

The reference period, spatial scale, temporal aggregation, and physical
units of $z_k(y)$ must be stated whenever $\beta_{ij}$ is reported.
Moreover, the graph restricts the set of scientifically admissible
directions but does not by itself identify causal effects. In the absence
of defensible identification assumptions, $w_{ij,k}(y)$ is interpreted
as a physics-constrained predictive triggering probability rather than
an intervention effect.

Given these edge-specific hazards, the state of each node is represented by an occurrence indicator and a conditional severity. For each node $v_j$, let $Z_j^{(r)}(y)\in\{0,1\}$ denote its onset status
at propagation step $r$. The value $Z_j^{(r)}(y)=1$ indicates that the
peril has occurred within the modeled event window. Conditional on
occurrence, $S_j^{(r)}(y)\in[0,1]$ denotes its normalized physical
severity. We impose $0\leq S_j^{(r)}(y)\leq Z_j^{(r)}(y)$, so that
severity is necessarily zero when no onset has occurred.

The vectors $\mathbf{Z}^{0}(y)$ and $\mathbf{S}^{0}(y)$ contain exogenous event-window shocks observed or simulated before propagation, with $0\leq\mathbf{S}^{0}(y)\leq\mathbf{Z}^{0}(y)$ component-wise. They may include root-node shocks and pre-existing background onsets at non-root nodes. To avoid double counting, a background mechanism represented in $(\mathbf{Z}^{0},\mathbf{S}^{0})$ is not simultaneously included in the background cumulative hazard $H_{0,j,k}(y)$. In the baseline numerical experiment, background occurrences are represented by $(\mathbf{Z}^{0},\mathbf{S}^{0})$ and $H_{0,j,k}(y)$ is set to zero.

Let $\mathcal{N}^{\mathrm{in}}_j
=\{i:(v_i,v_j)\in\mathcal{E}_{\mathrm{phys}}\}$ denote the direct
predecessors of node $v_j$. Define the edge cumulative hazard as
$H_{ij,k}(y)=-\log[1-w_{ij,k}(y)]$. For each edge, let
$q_{ij}:[0,1]\rightarrow[0,1]$ be a non-decreasing severity-response
function satisfying $q_{ij}(0)=0$ and $q_{ij}(1)=1$. A parsimonious
specification is $q_{ij}(s)=s^{\nu_{ij}}$ with $\nu_{ij}>0$.

Some target perils require two or more parent conditions to act jointly.
Let $\mathcal{I}_j$ contain the prespecified parent interactions relevant
to node $j$. For example, post-fire debris-flow onset may depend on the
joint presence of burn severity and intense precipitation. Let
$H^{\mathrm{int}}_{i\ell,j,k}(y)\geq0$ denote the cumulative hazard
associated with the interaction $(i,\ell)$ and let
$q^{\mathrm{int}}_{i\ell,j}$ be its normalized severity-response
function.

Conditional on the current parent states, the total onset probability
of a non-root node is
\begin{equation}
\begin{aligned}
\pi_{j,k}^{(r+1)}(y)
=
1-\exp\Bigg\{
&-H_{0,j,k}(y)
-\sum_{i\in\mathcal{N}^{\mathrm{in}}_j}
Z_i^{(r)}(y)
q_{ij}\!\left(S_i^{(r)}(y)\right)
H_{ij,k}(y)
\\
&-\sum_{(i,\ell)\in\mathcal{I}_j}
Z_i^{(r)}(y)Z_\ell^{(r)}(y)
q^{\mathrm{int}}_{i\ell,j}
\!\left(S_i^{(r)}(y),S_\ell^{(r)}(y)\right)
H^{\mathrm{int}}_{i\ell,j,k}(y)
\Bigg\},
\end{aligned}
\label{eq:onset-probability}
\end{equation}
where $H_{0,j,k}(y)\geq0$ is the background cumulative hazard not
attributed to the modeled parent nodes.

Equation~\eqref{eq:onset-probability} combines background, individual
parent, and prespecified interaction contributions on the cumulative
hazard scale. It therefore remains within the unit interval without
requiring the parent channels to be treated as mutually exclusive.
The edge probability $w_{ij,k}(y)$ affects occurrence, whereas the
conditional physical severity is modeled separately below.

Occurrence and physical severity are deliberately kept distinct. Once a node is exposed to its active parents and local physical drivers, its conditional severity is determined through a normalized stress response. For each non-root node, define a normalized physical stress index
\begin{equation}
X_{j,k}^{(r)}(y)
=
G_{j,k}
\left(
\{Z_i^{(r)}(y),S_i^{(r)}(y):
i\in\mathcal{N}^{\mathrm{in}}_j\},
\mathbf{x}_{j,k}(y)
\right)
\in[0,1],
\label{eq:physical-stress}
\end{equation}
where $\mathbf{x}_{j,k}(y)$ contains the physical drivers required for
the target peril, such as rainfall intensity, burn severity, slope,
antecedent moisture, wind, ignition conditions, or soil properties.
The function $G_{j,k}$ is specified from the physical mechanism and may
contain both main effects and parent interactions.

Importantly, $X_{j,k}^{(r)}(y)$ is a normalized physical stress index,
not an onset probability. This distinction allows a threshold parameter
to retain a physical interpretation after the underlying covariates have
been transformed to a documented common scale.

Let $\sigma(x)=1/[1+\exp(-x)]$ denote the standard logistic function.
For each non-root node, define the normalized conditional-severity
response
\begin{equation}
\Psi_j(x)
=
\frac{
\sigma\!\left(k_j(x-\theta_j)\right)
-
\sigma\!\left(-k_j\theta_j\right)
}{
\sigma\!\left(k_j(1-\theta_j)\right)
-
\sigma\!\left(-k_j\theta_j\right)
},
\qquad x\in[0,1],
\label{eq:severity-response}
\end{equation}
where $\theta_j\in(0,1)$ is the threshold on the normalized physical
stress scale and $k_j>0$ controls the steepness of the transition.

The normalization gives $\Psi_j(0)=0$, $\Psi_j(1)=1$, and
$0\leq\Psi_j(x)\leq1$. Exact boundary anchoring is a convenient
normalization for the controlled model, not a theoretical necessity:
Lemma~\ref{lem:monot} through Theorem~\ref{thm:certified-bound} require only
that the severity response be non-decreasing and map into $[0,1]$. The
parameters $\theta_j$ and $k_j$ must be estimated using the same
transformation and physical covariates used to construct
$X_{j,k}^{(r)}(y)$. They must not be calibrated on the triggering-probability
scale.

The edge hazards, occurrence states, and severity response together define the realized cascade. Let $U_{j,k}\in(0,1)$ be a scenario-level latent uniform variable used
to realize the onset of non-root node $v_j$. The same draw is retained
throughout the propagation steps of a given scenario. Initialize
$(Z_j^{(0)}(y),S_j^{(0)}(y))=(Z_j^{0}(y),S_j^{0}(y))$. Conditional on
the exogenous shocks, climate state, spatial unit, and latent uniforms,
each non-root node evolves according to
\begin{equation}
\begin{aligned}
Z_j^{(r+1)}(y)
&=
\max\left\{
Z_j^{(r)}(y),
\mathbf{1}\!\left[U_{j,k}<\pi_{j,k}^{(r+1)}(y)\right]
\right\},
\\
S_j^{(r+1)}(y)
&=
\max\left\{
S_j^{(r)}(y),
Z_j^{(r+1)}(y)
\Psi_j\!\left(X_{j,k}^{(r)}(y)\right)
\right\}.
\end{aligned}
\label{eq:cascade}
\end{equation}

For a root node, the endogenous onset probability is set to zero and its
state remains equal to the exogenous pair $(Z_j^{0}(y),S_j^{0}(y))$.
For a non-root node, the first line of Equation~\eqref{eq:cascade}
realizes onset and makes persistence explicit: once $Z_j^{(r)}=1$, it
cannot return to zero. The second line analogously preserves previously
realized severity. Under the initialization above and the monotonicity
conditions used below, this recursion is equivalent to taking the maximum
against the exogenous state at every iteration, but it displays the
persistence property directly.

The maximum operators preserve an event or severity once it has occurred
within the modeled event window. If the functions
$q_{ij}$, $q^{\mathrm{int}}_{i\ell,j}$, and $G_{j,k}$ are
coordinate-wise non-decreasing and all cumulative hazard contributions
are non-negative, the realized cascade operator is order-preserving
conditional on the latent uniforms. On a finite event-scale DAG, the
states can consequently be evaluated in topological order and stabilize
after a finite number of propagation steps.

Residual dependence among the latent uniforms may be introduced through
a separately specified dependence model when it is scientifically
required. Such residual dependence represents common causes not captured
by the observed covariates and must not duplicate the directional
dependence already represented by the cascade edges.

\subsection{Financial Loss and Contractual Aggregation}
\label{subsec:financial-contract-mapping}

After the event-scale cascade has reached closure, the physical states are translated into component losses and then into the declared reinsurance contract. Let $\mathbf{S}^{*}(y)$ denote the final realized severity vector. Since
$S_j^{*}(y)=0$ whenever $Z_j^{*}(y)=0$, occurrence and conditional
severity are both reflected in the financial mapping.

Let $a_m\geq0$ be an exposure weight specified before loss-model
estimation, with $\sum_{m=1}^{n}a_m=1$. The exposure-weighted physical
footprint is
\begin{equation}
C(\mathbf{S}^{*})
=
\sum_{m=1}^{n}a_mS_m^{*},
\qquad
0\leq C(\mathbf{S}^{*})\leq1.
\label{eq:systemic-footprint}
\end{equation}

For component $j$, let $\lambda_j>0$ denote its maximum modeled
ground-up loss and let $\gamma_j\geq1$ control the baseline vulnerability
curve. Define the baseline damage fraction as
$d_j(\mathbf{S}^{*})=(S_j^{*})^{\gamma_j}$ and the demand-surge
multiplier as $B_j(\mathbf{S}^{*})
=\exp\{\alpha_jC(\mathbf{S}^{*})\}$, where $\alpha_j\geq0$.

The component-level ground-up loss is
\begin{equation}
L_j(\mathbf{S}^{*})
=
\lambda_j
\frac{
B_j(\mathbf{S}^{*})d_j(\mathbf{S}^{*})
}{
1+
\left[B_j(\mathbf{S}^{*})-1\right]
d_j(\mathbf{S}^{*})
}.
\label{eq:bounded-loss}
\end{equation}

Equation~\eqref{eq:bounded-loss} applies the demand-surge adjustment on
the odds scale. It satisfies $L_j(\mathbf{0})=0$ and
$0\leq L_j(\mathbf{S}^{*})\leq\lambda_j$, so demand surge cannot produce
an unbounded loss. The parameter $\alpha_j$ should be interpreted as a
repair-market amplification parameter only when the modeled losses
occur within a common period of material, labour, and accommodation
scarcity.

A cascade does not automatically constitute a single contractual
occurrence. Let $\mathcal{A}$ denote a declared contractual aggregation
unit and let $\mathcal{E}(\mathcal{A})$ contain the covered events assigned
to that unit. The indicator $\chi_{j,e}$ equals one when component loss
$j$ from event $e$ is covered and allocated to $\mathcal{A}$, and equals
zero otherwise. The aggregated covered loss and the corresponding
single-layer ceded loss are
\begin{equation}
\begin{aligned}
L_{\mathcal{A}}
&=
\sum_{e\in\mathcal{E}(\mathcal{A})}
\sum_{j=1}^{n}
\chi_{j,e}L_{j,e},
\\
L_R^{\mathcal{A}}
&=
\min\left\{
\left(L_{\mathcal{A}}-D_{\mathcal{A}}\right)_+,
M_{\mathcal{A}}
\right\},
\qquad
(x)_+=\max\{x,0\}.
\end{aligned}
\label{eq:contract-loss}
\end{equation}

For an annual aggregate treaty, let $N_y$ be the number of contractually eligible event windows in year $y$, and let $L_{j,e,y}$ be component $j$'s loss from event window $e$. The annual ground-up and ceded losses used in the numerical study are
\begin{equation}
L_{\mathrm{ann},y}
=
\sum_{e=1}^{N_y}\sum_{j=1}^{n}\chi_{j,e,y}L_{j,e,y},
\qquad
L_{R,y}
=
\min\left\{(L_{\mathrm{ann},y}-D)_+,M\right\}.
\label{eq:annual-contract-loss}
\end{equation}
In the controlled numerical design, an eligible event window is an abstract,
non-overlapping aggregation unit whose maximum horizon is the largest declared
edge-lag horizon in the event-scale DAG. Every physical component loss is
allocated to exactly one window, windows do not overlap, and a wildfire and its
modeled downstream post-fire debris-flow component are assigned to the same
window. This one-to-one allocation prevents a component loss from being counted
again in another window before annual aggregation. This is the declared synthetic allocation convention for the annual
aggregate experiment. For an occurrence-cover application, the allocation rule
is replaced by explicit causation and hours-clause provisions.

The frequency law of $N_y$ is part of the actuarial model and cannot be replaced silently by a single-cascade calculation. In the controlled experiment, $N_y$ is Poisson with mean 1.50; this is a prespecified synthetic assumption rather than an empirical frequency estimate.

For an occurrence cover, $\mathcal{A}$ must be defined using the
contractual causation, hours, territorial, and covered-peril clauses.
Cascade components that do not satisfy the same occurrence definition
must be treated as separate occurrences. For an annual aggregate cover,
$\mathcal{A}$ may instead represent the contract year. Reinstatements,
aggregate deductibles, and multiple layers must be incorporated when
they are present in the actual treaty.

In an empirical implementation, the parameters $\lambda_j$, $\gamma_j$,
and $\alpha_j$ must be estimated from exposure, vulnerability, and
repair-cost data available within the training sample. The controlled
numerical experiment in this study uses prespecified synthetic values
and is intended for internal model verification rather than empirical
pricing validation.

\subsection{Analytical Properties of the Cascading Dynamics}
\label{ana}

Throughout this section, the climate evaluation time $y$, spatial unit
$k$, graph topology, and scenario-level latent uniforms are fixed.
Climate non-stationarity enters through changes in the triggering
hazards and physical drivers across calendar periods or climate
scenarios, rather than through changes during the propagation
iterations of a single event.

Define the admissible state space as
$\mathfrak{X}=\{(\mathbf{z},\mathbf{s})\in
\{0,1\}^n\times[0,1]^n:
\mathbf{0}\leq\mathbf{s}\leq\mathbf{z}\}$. For two admissible states,
we write $(\mathbf{z},\mathbf{s})\preceq
(\widetilde{\mathbf{z}},\widetilde{\mathbf{s}})$ when both
$\mathbf{z}\leq\widetilde{\mathbf{z}}$ and
$\mathbf{s}\leq\widetilde{\mathbf{s}}$ component-wise.

Let $\mathcal{T}_{y,k,\mathbf{u}}$ denote the realized cascade operator
defined by the right-hand side of Equation~\eqref{eq:cascade} for a fixed
latent-uniform vector $\mathbf{u}\in(0,1)^n$. We impose the following
regularity conditions throughout this section:

\begin{enumerate}
    \item all background, edge, and interaction cumulative hazards are
    non-negative;

    \item the functions $q_{ij}$ and
    $q^{\mathrm{int}}_{i\ell,j}$ are coordinate-wise non-decreasing;

    \item each physical stress function $G_{j,k}$ is coordinate-wise
    non-decreasing in its risk-increasing parent states and physical
    drivers;

    \item each conditional-severity function $\Psi_j$ is non-decreasing
    and maps $[0,1]$ into $[0,1]$.
\end{enumerate}

These conditions are restrictions on the orientation of the model
inputs. A physical covariate that reduces risk, such as soil moisture
under a particular definition, must either enter with its appropriate
sign or be transformed into a risk-increasing quantity before the
monotonicity results are applied.

\begin{lemma}
\label{lem:monot}

For any valid exogenous state
$(\mathbf{Z}^{0},\mathbf{S}^{0})\in\mathfrak{X}$ and any fixed
$\mathbf{u}\in(0,1)^n$, the realized cascade operator satisfies:

\begin{enumerate}
    \item $\mathcal{T}_{y,k,\mathbf{u}}$ maps
    $\mathfrak{X}$ into itself;

    \item if
    $(\mathbf{z},\mathbf{s})\preceq
    (\widetilde{\mathbf{z}},\widetilde{\mathbf{s}})$, then
    $\mathcal{T}_{y,k,\mathbf{u}}(\mathbf{z},\mathbf{s})
    \preceq
    \mathcal{T}_{y,k,\mathbf{u}}
    (\widetilde{\mathbf{z}},\widetilde{\mathbf{s}})$;

    \item the sequence generated by Equation~\eqref{eq:cascade} is
    non-decreasing under $\preceq$.
\end{enumerate}
\end{lemma}

\begin{proof}

The onset update is the maximum of two binary quantities and therefore
belongs to $\{0,1\}$. The conditional-severity response belongs to
$[0,1]$. Since the previous state is admissible,
$S_j^{(r)}\leq Z_j^{(r)}\leq Z_j^{(r+1)}$, and
$Z_j^{(r+1)}\Psi_j(X_{j,k}^{(r)})\leq Z_j^{(r+1)}$. It follows that
$0\leq S_j^{(r+1)}\leq Z_j^{(r+1)}$, so the updated state remains in
$\mathfrak{X}$.

Suppose
$(\mathbf{z},\mathbf{s})\preceq
(\widetilde{\mathbf{z}},\widetilde{\mathbf{s}})$. The non-negativity
of the cumulative hazards and the monotonicity of the functions
$q_{ij}$ and $q^{\mathrm{int}}_{i\ell,j}$ imply that the cumulative
hazard entering the target node cannot decrease when a parent
occurrence or severity increases. Because $x\mapsto1-\exp(-x)$ is
increasing, the corresponding onset probabilities satisfy
$\pi_j(\mathbf{z},\mathbf{s})
\leq\pi_j(\widetilde{\mathbf{z}},\widetilde{\mathbf{s}})$.

For fixed $u_j$, the indicator
$\mathbf{1}\{u_j<\pi_j\}$ is non-decreasing in $\pi_j$. Hence the
updated occurrence indicator is order-preserving. The assumed
monotonicity of $G_{j,k}$ and $\Psi_j$ then implies that the updated
conditional severity is also order-preserving. This proves the second
claim.

Finally, both component updates explicitly take the maximum with the
previous iterate. Hence
$(\mathbf{Z}^{(r)},\mathbf{S}^{(r)})\preceq
(\mathbf{Z}^{(r+1)},\mathbf{S}^{(r+1)})$ at every step. This directly
establishes persistence and the third claim.
\end{proof}

\begin{theorem}
\label{thm:exist}

Let $\mathcal{G}=(\mathcal{V},\mathcal{E}_{\mathrm{phys}})$ be a finite
directed acyclic graph. Fix the climate state $y$, spatial unit $k$,
exogenous state $(\mathbf{Z}^{0},\mathbf{S}^{0})\in\mathfrak{X}$,
model parameters, physical drivers, and latent-uniform vector
$\mathbf{u}\in(0,1)^n$.

Then the sequence generated by Equation~\eqref{eq:cascade} reaches a
unique realized cascade closure
$(\mathbf{Z}^{*},\mathbf{S}^{*})$ after finitely many propagation
steps. If $H$ is the maximum number of edges in any directed path of
the graph, then
\begin{equation}
(\mathbf{Z}^{(H)},\mathbf{S}^{(H)})
=
(\mathbf{Z}^{(H+1)},\mathbf{S}^{(H+1)})
=
(\mathbf{Z}^{*},\mathbf{S}^{*}).
\label{eq:finite-convergence}
\end{equation}

The closure is unique conditional on the exogenous shocks, physical
drivers, model parameters, graph topology, and latent uniforms.
Consequently, a specified joint distribution for these random inputs
induces a unique probability distribution for the cascade closure. This
is uniqueness of deterministic forward evaluation once all stochastic
primitives are fixed; it does not imply statistical identifiability of
parameters from observed data or observational equivalence of competing
model specifications.
\end{theorem}

\begin{proof}

Because the graph is a finite DAG, its nodes admit a topological
ordering. Let $d(j)$ denote the maximum number of directed edges in a
path terminating at node $v_j$. Root nodes have depth zero, and
$H=\max_j d(j)$.

A root node has no endogenous parents. By construction, its occurrence
and severity states remain equal to
$(Z_j^0,S_j^0)$ at every propagation step. Hence all root states are
stable at the initial step.

Suppose that every node with depth at most $m-1$ has reached its final
state by iteration $m-1$. Every parent of a node with depth $m$ has
depth at most $m-1$. Therefore, all quantities entering its onset
probability and physical stress function are fixed from iteration
$m-1$ onward. The fixed latent uniform then uniquely determines its
onset indicator, and the realized onset together with the fixed physical
stress uniquely determines its conditional severity. Thus every node
of depth $m$ stabilizes no later than iteration $m$.

Induction over node depth proves Equation~\eqref{eq:finite-convergence}.
Uniqueness follows from the same topological argument. The state of
each root is uniquely fixed by the exogenous inputs. Once the states of
all predecessors of a non-root node are fixed, its onset probability,
onset indicator, physical stress, and conditional severity are uniquely
determined. Proceeding through a topological ordering uniquely
determines every component of the realized closure.
\end{proof}

\begin{corollary}
\label{cor:comparative-statics}

Consider two cascade specifications on the same DAG with the same
latent-uniform vector, response functions, and graph topology. Suppose
that the second specification has component-wise larger exogenous
occurrence indicators and severities, edge probabilities, background
hazards, interaction hazards, and risk-increasing physical drivers.

Then its realized closure is no smaller than the closure of the first
specification under the order $\preceq$.
\end{corollary}

\begin{proof}

Under the stated ordering, the target onset probabilities and physical
stress indices in the second specification are no smaller at every
node. The result follows by induction over a topological ordering using
Lemma~\ref{lem:monot}.
\end{proof}

When $\beta_{ij}\geq0$ in Equation~\eqref{eq:dynamic-edge-weight}, an
increase in the climate anomaly weakly increases the corresponding
edge probability while other covariates are held fixed. Therefore,
under the assumptions of Corollary~\ref{cor:comparative-statics}, the
common-random-number closure is non-decreasing in the climate anomaly.
This is a conditional model sensitivity and not an empirically
identified causal effect of temperature.

\begin{remark}
\label{rem:chaos}

Finite-step closure is a computational and mathematical property and
does not imply that the final state is economically moderate. Using the
exposure-weighted footprint in
Equation~\eqref{eq:systemic-footprint}, a systemic cascade may be
declared when $C(\mathbf{S}^{*})\geq\tau_{\mathrm{cat}}$, where
$\tau_{\mathrm{cat}}\in(0,1)$ is specified before evaluating the test
sample.

Because $S_j^{*}=0$ when $Z_j^{*}=0$, the footprint reflects both the
number of realized peril components and their conditional severities.
Equal exposure weights recover the normalized aggregate-severity
measure $n^{-1}\|\mathbf{S}^{*}\|_1$. Catastrophic behavior is therefore
defined through the realized cascade closure, not through
non-convergence or a basin-of-attraction argument.
\end{remark}

\begin{remark}
\label{rem:dag-scope}

The finite-step and uniqueness results are stated for an event-scale DAG.
For networks containing within-event feedback cycles, closure can instead be
defined under a contraction condition, a least-fixed-point selection rule, or
an explicit temporal expansion that represents delayed feedback on an acyclic
time-indexed graph.
\end{remark}

\begin{theorem}
\label{thm:certified-bound}

Fix the graph topology, response functions, coverage rules, and contractual aggregation map. Let $\boldsymbol{\Xi}$ collect the aleatory inputs, including the annual event count, exogenous event-window shocks, and latent onset uniforms, and fix its joint law. Let $\boldsymbol{\eta}$ collect only ordered scenario or stress inputs whose increase cannot reduce an onset probability, physical stress, or financial loss; examples include edge and interaction hazards, risk-increasing physical drivers, loss capacities, and nonnegative demand-surge parameters.

Suppose that these ordered inputs belong to the rectangular stress set
\begin{equation}
\mathcal{U}
=
\left\{
\boldsymbol{\eta}:
\underline{\boldsymbol{\eta}}
\preceq
\boldsymbol{\eta}
\preceq
\overline{\boldsymbol{\eta}}
\right\}.
\label{eq:uncertainty-set}
\end{equation}
For any fixed realization $\boldsymbol{\xi}$ of $\boldsymbol{\Xi}$, evaluate every admissible specification with the same exogenous shocks and latent uniforms. Then
\begin{equation}
L_{\mathcal{A}}(\boldsymbol{\eta},\boldsymbol{\xi})
\leq
L_{\mathcal{A}}(\overline{\boldsymbol{\eta}},\boldsymbol{\xi}),
\qquad
L_R^{\mathcal{A}}(\boldsymbol{\eta},\boldsymbol{\xi})
\leq
L_R^{\mathcal{A}}(\overline{\boldsymbol{\eta}},\boldsymbol{\xi})
\label{eq:pathwise-corner-bound}
\end{equation}
for every $\boldsymbol{\eta}\in\mathcal{U}$. Consequently,
\begin{equation}
\begin{aligned}
\sup_{\boldsymbol{\eta}\in\mathcal{U}}
\mathbb{E}_{\boldsymbol{\Xi}}[L_{\mathcal{A}}(\boldsymbol{\eta},\boldsymbol{\Xi})]
&=
\mathbb{E}_{\boldsymbol{\Xi}}[L_{\mathcal{A}}(\overline{\boldsymbol{\eta}},\boldsymbol{\Xi})],\\
\sup_{\boldsymbol{\eta}\in\mathcal{U}}
\mathbb{E}_{\boldsymbol{\Xi}}[L_R^{\mathcal{A}}(\boldsymbol{\eta},\boldsymbol{\Xi})]
&=
\mathbb{E}_{\boldsymbol{\Xi}}[L_R^{\mathcal{A}}(\overline{\boldsymbol{\eta}},\boldsymbol{\Xi})].
\end{aligned}
\label{eq:expected-corner-bound}
\end{equation}
\end{theorem}

\begin{proof}
For one event window, Corollary~\ref{cor:comparative-statics} gives component-wise ordering of the closure under common aleatory inputs. The footprint and the bounded loss map are non-decreasing in every severity component. In particular, with $d_j=(S_j^*)^{\gamma_j}$ and $B_j=\exp\{\alpha_jC(\mathbf{S}^*)\}$, the adjusted damage fraction $q_j=B_jd_j/[1+(B_j-1)d_j]$ satisfies $\partial q_j/\partial d_j>0$ and $\partial q_j/\partial C=\alpha_jq_j(1-q_j)\geq0$. Summing covered losses over a fixed realized number of event windows preserves the ordering, and the layer map $x\mapsto\min\{(x-D)_+,M\}$ is non-decreasing. The pathwise inequalities therefore hold for every $\boldsymbol{\xi}$. Taking expectations preserves them, and the upper corner belongs to $\mathcal{U}$, so it attains the suprema.
\end{proof}

\begin{remark}
\label{rem:bound-scope}
Theorem~\ref{thm:certified-bound} is a conditional stress-test result. It does not place aleatory occurrence indicators at their upper values, convert a joint confidence or posterior region into an independent rectangle, or optimize over an uncertain residual-dependence law. If physical or statistical constraints make the upper-corner combination infeasible, the rectangular result remains conservative but may be unattainable. Correlated feasible sets require constrained or distributionally robust optimization rather than a single corner evaluation.
\end{remark}

\begin{remark}
\label{rem:constrained-corner}
The rectangular corner can be actively misleading when ordered stresses
cannot attain their marginal maxima simultaneously. Consider two ordered
stress inputs $(\eta_1,\eta_2)\in[0,1]^2$ and the monotone loss proxy
$\ell(\boldsymbol{\eta})=\eta_1+\eta_2$. Over the rectangular envelope,
the upper-corner evaluation is $\ell(1,1)=2$. If the scientifically
feasible set is instead
\[
\mathcal{U}_{c}=\{(\eta_1,\eta_2)\in[0,1]^2:\eta_1+\eta_2\leq1.2\},
\]
then $(1,1)$ is infeasible and the true constrained maximum is only $1.2$.
Thus, evaluating the enclosing rectangle overstates the tight stress bound
by $0.8$ in this example. A tight bound over $\mathcal{U}_c$ requires a
constrained optimization, even though the loss map remains monotone.
\end{remark}

\begin{corollary}
\label{cor:universal-loss-bound}

For any realized cascade and any fixed contractual aggregation unit
containing one modeled event, the component losses satisfy
$0\leq L_j\leq\lambda_j$. Consequently,
$0\leq L_{\mathcal{A}}\leq\sum_{j=1}^{n}\lambda_j$ and
$0\leq L_R^{\mathcal{A}}
\leq\min\{(\sum_j\lambda_j-D_{\mathcal{A}})_+,
M_{\mathcal{A}}\}\leq M_{\mathcal{A}}$.
\end{corollary}

For an aggregation unit containing multiple covered events, the
corresponding universal bound must also account for the number of events,
event limits, aggregate limits, and any contractual reinstatements.

\section{Estimation and Data Protocol}
\label{data}

This section specifies the data construction, estimation, validation,
and scenario propagation procedures required for an empirical
implementation of the CCRN. The numerical results reported in the
present paper are based on controlled synthetic experiments.
Accordingly, the procedures below define a future empirical
implementation protocol and must not be interpreted as evidence that
the current model parameters have already been estimated from observed
California insurance claims.

\subsection{Core Occurrence, Severity, and Scenario Model}
\label{subsec:estimation-aggregation}

An empirical implementation first prespecifies the event-scale topology,
risk sets, lag windows, event definitions, and contractual aggregation rule.
The detailed product-construction, harmonization, and training-only workflow
is given in Appendix~\ref{app:data-workflow}. The main text retains the core
statistical equations needed to connect the empirical model to the cascade.

With the event states and risk sets frozen, occurrence and severity are estimated as two linked parts of the same node-level model. For each target node, all admissible parent contributions and
prespecified parent interactions are included jointly. Conditional on
the target being at risk, the discrete-time onset probability is
modeled as
\begin{equation}
\begin{aligned}
\Pr(Y_{j,k,t}=1\mid R_{j,k,t}=1)
=
1-\exp\Bigg\{
&-H_{0,j,k,t}
-\sum_{i\in\mathcal{N}^{\mathrm{in}}_j}
Q_{ij,k,t}H_{ij,k,t}
\\
&-\sum_{(i,\ell)\in\mathcal{I}_j}
Q^{\mathrm{int}}_{i\ell,j,k,t}
H^{\mathrm{int}}_{i\ell,j,k,t}
\Bigg\}.
\end{aligned}
\label{eq:joint-trigger-hazard}
\end{equation}

The background cumulative hazard $H_{0,j,k,t}$ represents target onsets
not attributed to the modeled parents. The edge cumulative hazard is
specified as
$H_{ij,k,t}=H_{ij,k}^{0}
\exp\{\beta_{ij}z_k(t)
+\boldsymbol{\gamma}_{ij}^{\top}\mathbf{c}_{k,t}\}$, and its
corresponding probability contribution is
$w_{ij,k,t}=1-\exp(-H_{ij,k,t})$. This is algebraically consistent with
Equation~\eqref{eq:dynamic-edge-weight}.

The interaction cumulative hazard may have its own climate and physical
covariates. For example, the wildfire and rainfall interaction for
post-fire debris flow may depend on burn severity, rainfall intensity,
slope, soil type, time since fire, and antecedent precipitation.
Interaction terms are specified from the physical mechanism rather than
selected through an unrestricted search over all possible products.

Main-effect and interaction hazards can nevertheless be weakly identifiable,
especially when joint parent activation is rare. An empirical fit should impose
a hierarchy rule under which an interaction is freely estimated only when the
corresponding main effects are retained, center interacting covariates before
forming products, and use shrinkage toward zero for the interaction coefficient.
The combined target-node hazard is the primary estimable and validated quantity;
its decomposition into main and interaction channels should not be overinterpreted
when posterior or bootstrap uncertainty indicates weak separation.

Spatial heterogeneity is represented through hierarchical random
effects or region-specific partial pooling. Seasonal terms and known
synoptic drivers are included when required by the target process.
Remaining temporal dependence is examined through residual diagnostics
and accommodated using temporal random effects, cluster-robust
inference, or block resampling. First-order differencing is not treated
as sufficient evidence of stationarity or causal identification.

Occurrence-model validation should report effective onset counts,
exposure time in the risk set, calibration curves, Brier or log scores,
and appropriate residual diagnostics. Reporting coefficient
significance without predictive calibration is insufficient, especially
when target onsets are rare.

The occurrence model determines whether a new target onset is realized; the second stage models its physical magnitude. Conditional severity is estimated only for observations with a realized
target onset. For each observed onset, the physical stress index is
constructed using the same function $G_{j,k}$ and the same physical
covariates specified in Equation~\eqref{eq:physical-stress}. The
conditional mean severity is linked to physical stress through the
normalized response $\Psi_j$ in
Equation~\eqref{eq:severity-response}.

For severities strictly between zero and one, a hierarchical beta model
may be written as
\begin{equation}
S_j^k(t)\mid
\{Y_{j,k,t}=1,X_{j,k,t}\}
\sim
\operatorname{Beta}
\left(
\mu_{j,k,t}\phi_j,
[1-\mu_{j,k,t}]\phi_j
\right),
\qquad
\mu_{j,k,t}=\Psi_j(X_{j,k,t}),
\label{eq:conditional-severity-model}
\end{equation}
where $\phi_j>0$ controls conditional dispersion. If exact boundary
values zero or one occur after normalization, a zero-one-inflated beta
model or another bounded response distribution is required.

Equation~\eqref{eq:conditional-severity-model} is the observational model
used to estimate the conditional mean response $\Psi_j$ and its dispersion.
The realized cascade recursion propagates the fitted conditional mean
$\Psi_j(X_{j,k,t})$, while the beta innovation represents residual measurement
and local-heterogeneity variation around that mean. This separation preserves
the deterministic forward cascade conditional on its fitted inputs and keeps
the closure and order-preservation results aligned with the numerical
implementation.

Exact anchoring at $0$ and $1$ may also be relaxed in empirical work by using a
flexible monotone spline or another bounded monotone mean function. The analytical
results continue to hold provided the fitted response remains non-decreasing and
maps into $[0,1]$; exact endpoint equality is not required.

The parameters of $G_{j,k}$, the severity-response functions
$q_{ij}$, the physical threshold $\theta_j$, the steepness $k_j$, and
the dispersion parameter are estimated using training observations.
Externally documented physical ranges may be imposed as constraints or
prior distributions. They must not be selected after examining
test-sample losses.

Occurrence and conditional severity may be estimated jointly in a
hierarchical hurdle model or sequentially with full propagation of
first-stage uncertainty. A sequential point-estimation procedure that
treats the estimated occurrence parameters as known in the severity
stage understates parameter uncertainty.

Severity-model validation should report a proper distributional score,
calibration by predicted-severity bins, residual diagnostics, and
performance within physically relevant tail regions. Validation based
only on mean squared error is insufficient for a reinsurance
application.

After the node-level occurrence and severity models have been fitted, their outputs are translated into spatially consistent monetary losses before contract terms are applied. Region-specific cascades are evaluated using region-specific hazards,
physical stress functions, and exposure definitions. Cell-specific
thresholds and vulnerabilities are not replaced by simple
total-insured-value-weighted averages before cascade evaluation.

If $L_{\mathcal{A},k}$ is already expressed as a monetary loss for
spatial unit $k$, portfolio loss is obtained by direct summation. If
$\ell_{\mathcal{A},k}$ is instead a normalized loss rate and
$V_k$ is the corresponding insured value, exposure weighting is applied
once:
\begin{equation}
L_{\mathcal{A}}^{\mathrm{portfolio}}
=
\sum_{k=1}^{K}L_{\mathcal{A},k}
=
\sum_{k=1}^{K}V_k\ell_{\mathcal{A},k}.
\label{eq:regional-loss-aggregation}
\end{equation}

Applying normalized exposure weights to losses that are already in
monetary units would double-count the exposure scaling. Contractual
deductibles, limits, occurrence definitions, and annual aggregation
rules are applied at the level specified by the treaty rather than
independently to each spatial cell unless the contract explicitly
requires such treatment.

Future climate scenarios enter only after the historical specification and all transformations have been frozen. The climate-sensitivity parameter $\beta_{ij}$ is estimated from
historical observations and is not estimated by treating future
climate-model simulations as observed outcomes. After historical model
estimation, bias-adjusted climate-model output provides future values
of $z_k(y)$ and other required physical covariates. These values are
inserted into Equation~\eqref{eq:dynamic-edge-weight} and the relevant
physical stress functions without re-estimating the historical
coefficients.

The reference period, anomaly definition, downscaling method, and
bias-adjustment procedure must be reported. Trend-preserving methods
are preferred when the quantity of interest is a future change in
hazard. Multiple climate models, ensemble members, and emissions
scenarios are retained rather than collapsed into a single deterministic
trajectory.

Parameter uncertainty, internal climate variability, climate model
uncertainty, emissions scenario uncertainty, and exposure uncertainty
are recorded separately before being combined in the loss simulation.
A single high-emissions pathway is reported as a conditional stress
scenario and not as a unique forecast.

The complete training-only estimation and climate-scenario propagation
workflow is summarized in Algorithm~\ref{alg:calibration}. Its placement at
the end of this section consolidates the occurrence, severity, spatial-loss,
validation, uncertainty, and climate-scenario steps developed above.

\begin{algorithm}[H]
\caption{Training-Only Estimation and Climate-Scenario Propagation}
\label{alg:calibration}
\footnotesize
\begin{algorithmic}[1]
\Require Fixed event-scale DAG and interaction sets; lag windows, event and contract definitions; historical environmental, event, and exposure data; climate-model ensemble; frozen training, validation, and test dates.
\Ensure Fitted occurrence and severity models; frozen transformations; scenario-specific cascade inputs; uncertainty distributions; out-of-sample diagnostics.
\State Prespecify nodes, admissible edges, interactions, risk sets, physical covariates, lag windows, event definitions, and the contractual aggregation unit.
\State Freeze the data-split dates before transformation, lag selection, model selection, or tuning.
\State Harmonize the input products to documented spatial supports while preserving the temporal resolution, quality flags, missingness indicators, and version metadata required by each physical mechanism.
\State Estimate normalization and cross-product calibration rules on the training sample only, then apply the frozen rules to validation, test, and scenario data.
\State Construct occurrence and onset indicators, risk sets, severity-modulated parent contributions, and prespecified interaction contributions.
\For{each target node $j\in\mathcal{V}$}
    \State Fit the joint occurrence-hazard model in Equation~\eqref{eq:joint-trigger-hazard} using all admissible parent and interaction channels.
    \State Fit Equation~\eqref{eq:conditional-severity-model} to realized target onsets and estimate spatial heterogeneity by hierarchical partial pooling.
    \State Retain the joint covariance matrix, bootstrap distribution, or posterior distribution of the fitted parameters.
    \State Evaluate occurrence and severity calibration, proper scores, tail diagnostics, and residual dependence on the validation sample; freeze the selected specification before test evaluation.
\EndFor
\State Translate region-specific cascade outputs into monetary losses, aggregate them using Equation~\eqref{eq:regional-loss-aggregation}, and then apply the declared contract mapping.
\State Quantify parameter and exposure uncertainty through coherent block bootstrap refits or draws from the joint hierarchical posterior.
\For{each climate model, ensemble member, and emissions scenario}
    \State Bias-adjust the required scenario covariates using a trend-preserving rule estimated over the historical reference period.
    \State Insert the projected covariates into Equation~\eqref{eq:dynamic-edge-weight} and the physical stress functions without refitting historical response coefficients.
    \State Propagate parameter, climate-model, internal-variability, and exposure uncertainty through the cascade, financial-loss map, and contractual aggregation.
\EndFor
\State Report sample sizes, event counts, parameter uncertainty, calibration diagnostics, scenario definitions, and out-of-sample contract-level results.
\end{algorithmic}
\end{algorithm}

The controlled numerical study implements this workflow in a synthetic
recovery experiment and propagates the resulting parameter distributions through
the cascade, financial-loss map, and annual aggregate contract.

\section{Numerical Study}
\label{num}

The numerical study is a controlled verification of the CCRN architecture.
All inputs are prespecified synthetic quantities. The main comparison uses
$350{,}000$ contract years, while the matched event-level marginal sample
contains $500{,}000$ cascade realizations. Sensitivity experiments use separate
Monte Carlo samples and the sample sizes stated below.

\subsection{Annual Design and Matched-Marginal Benchmarks}
\label{subsec:annual-design-comparison}

The event-scale graph contains fuel aridity ($v_1$), wildfire ($v_2$), extreme
precipitation or flood ($v_3$), and post-fire debris flow ($v_4$). Fuel aridity
and precipitation are exogenous drivers. Wildfire may arise from an exogenous
background onset or through the fuel-aridity edge. Debris flow may arise from a
background onset or through the joint activation of wildfire, precipitation,
and terrain. Fuel aridity carries no direct insured loss.

For each contract year, the baseline number of eligible event windows satisfies
\[
N_y\sim\operatorname{Poisson}(1.50).
\]
Conditional on $N_y$, event windows are independent and climate inputs are held
fixed within the contract year. Each component loss is assigned to one and only
one event window before annual aggregation. The wildfire and its associated
post-fire debris-flow response are treated as components of the same synthetic
recovery window, so the allocation rule prevents double counting.

\begin{table}[H]
\centering
\small
\caption{Prespecified event-window inputs for the controlled CCRN. The
parameter $q_j$ is the exogenous pre-cascade occurrence probability, and loss
capacities are in USD billions.}
\label{tab:synthetic-design}
\begin{tabular}{@{}lccccc@{}}
\toprule
\textbf{Node} & $\boldsymbol{q_j}$ & \textbf{Conditional severity} & $\boldsymbol{\lambda_j}$ & $\boldsymbol{\gamma_j}$ & $\boldsymbol{\alpha_j}$ \\
\midrule
Fuel aridity & 0.300 & $\operatorname{Beta}(2.2,3.0)$ & 0.0 & 1.00 & 0.00 \\
Wildfire & 0.015 & $\operatorname{Beta}(1.8,4.0)$ & 3.8 & 1.45 & 0.30 \\
Extreme precipitation or flood & 0.100 & $\operatorname{Beta}(1.8,3.8)$ & 2.4 & 1.55 & 0.20 \\
Post-fire debris flow & 0.001 & $\operatorname{Beta}(1.6,4.5)$ & 1.8 & 1.65 & 0.40 \\
\bottomrule
\end{tabular}
\end{table}

Table~\ref{tab:synthetic-design} specifies the node-level inputs of the
controlled event-window design. The exogenous occurrence probabilities
$q_j$ determine background activation before propagation, the beta laws
describe conditional physical severity, and $\lambda_j$, $\gamma_j$, and
$\alpha_j$ determine the loss capacity, vulnerability curvature, and
demand-surge response. The table also shows that fuel aridity acts only as a
physical driver, whereas wildfire, precipitation, and debris flow contribute
directly to insured loss.

The reference triggering probabilities are $p_{12}^{0}=0.35$ and
$p_{(2,3)4}^{0}=0.55$, with climate coefficients $0.35$ and $0.18$ per degree
Celsius. Wind follows $\operatorname{Beta}(2.5,2.0)$ and terrain follows
$\operatorname{Beta}(2.8,1.8)$. The wildfire stress index is $S_1W$, whereas
the debris-flow interaction gate is $(S_2S_3T)^{1/3}$. The normalized logistic
severity responses use $(\theta_2,k_2)=(0.20,8)$ and
$(\theta_4,k_4)=(0.25,9)$. The footprint weights are
$(0,0.45,0.30,0.25)$, and the illustrative annual aggregate layer has
deductible $D=1.50$ and limit $M=0.50$ billion USD. These declarations make
every numerical input to Equations~\eqref{eq:dynamic-edge-weight}--\eqref{eq:annual-contract-loss} explicit.

At $\Delta T=1.2^\circ\mathrm{C}$, event-level CCRN severities are transformed
to pseudo-observations using average empirical ranks, including the common
mid-rank assigned to the point mass at zero. Because the severity margins contain
atoms at zero, the mid-rank transform is a modeling convention rather than a
unique probability-integral transform for the discontinuous margins. A randomized
distributional transform was not used in the present benchmark; sensitivity to
such alternative zero-mass transforms is therefore a limitation of the current
matched-marginal calibration. Normal scores determine the latent correlation
matrix, which is projected to the nearest correlation matrix when required.
Gaussian and Student-$t$ samples are mapped through the same empirical generalized
inverse marginals. The three annual models therefore share the marginal severity
distributions, annual frequency law, financial mapping, and contract terms, while
differing in their joint-dependence construction.

\begin{table}[H]
\centering
\small
\caption{Contract-level matched-marginal comparison. Parentheses contain
model-wise Monte Carlo standard errors.}
\label{tab:pricing-comparison}
\begin{tabular}{@{}lccc@{}}
\toprule
\textbf{Model} & \shortstack{\textbf{Attachment}\\\textbf{(\%)}} & \shortstack{\textbf{Ceded loss}\\\textbf{(USD million)}} & \shortstack{\textbf{Pure premium}\\\textbf{(\% of limit)}} \\
\midrule
Gaussian copula & 5.27 (0.04) & 22.47 (0.19) & 4.49 (0.04) \\
Student-$t$ copula ($\nu=5$) & 5.17 (0.04) & 22.45 (0.17) & 4.49 (0.03) \\
CCRN & 5.15 (0.04) & 22.29 (0.17) & 4.46 (0.03) \\
\bottomrule
\end{tabular}
\end{table}

Table~\ref{tab:pricing-comparison} compares the attachment probability,
expected ceded loss, and pure premium produced by the CCRN and the two
matched-marginal copula benchmarks. The three models give closely aligned
contract-level values and small Monte Carlo standard errors for the selected
layer. This indicates that, after matching the event-level marginals, the
attachment region of this low and relatively narrow layer is only weakly
sensitive to the remaining differences in joint-dependence shape.

\begin{table}[H]
\centering
\small
\caption{Tail metrics for annual aggregate ground-up loss $L_{\mathrm{ann}}$
and retained loss $L_{\mathrm{ret}}=L_{\mathrm{ann}}-L_R$, in USD billions.
Parentheses contain bootstrap standard errors.}
\label{tab:tail-comparison}
\begin{tabular}{@{}lcccc@{}}
\toprule
\textbf{Model} & \shortstack{$\boldsymbol{\operatorname{VaR}_{99.5\%}}$\\\textbf{Ground-up}} & \shortstack{$\boldsymbol{\operatorname{TVaR}_{99.5\%}}$\\\textbf{Ground-up}} & \shortstack{$\boldsymbol{\operatorname{VaR}_{99.5\%}}$\\\textbf{Retained}} & \shortstack{$\boldsymbol{\operatorname{TVaR}_{99.5\%}}$\\\textbf{Retained}} \\
\midrule
Gaussian copula & 3.701 (0.008) & 4.443 (0.028) & 3.201 (0.008) & 3.943 (0.028) \\
Student-$t$ copula ($\nu=5$) & 4.049 (0.022) & 5.016 (0.030) & 3.549 (0.022) & 4.516 (0.030) \\
CCRN & 3.748 (0.009) & 4.759 (0.034) & 3.248 (0.009) & 4.259 (0.034) \\
\bottomrule
\end{tabular}
\end{table}

Table~\ref{tab:tail-comparison} shows that the annual far-tail results separate
more clearly than the central contract metrics. The Student-$t$ construction
produces the largest ground-up and retained VaR and TVaR, the Gaussian
benchmark produces the smallest values, and the CCRN lies between them. The
reported tail standard errors are based on $B_{\mathrm{boot}}=60$ nonparametric
bootstrap replications obtained by i.i.d. resampling of complete simulated
contract years within each model. At the 99.5\% level, the ground-up quantiles
and all corresponding tail observations lie above the exhaustion point
$D+M=2.0$ billion in the reported experiment. Hence the retained loss equals
$L_{\mathrm{ann}}-M$ throughout the relevant tail, which explains why the
retained VaR and TVaR standard errors equal their ground-up counterparts. This
equality is specific to the exhausted-tail regime and is not a general identity
for the layer transformation.

\begin{figure}[H]
\centering
\includegraphics[width=\textwidth]{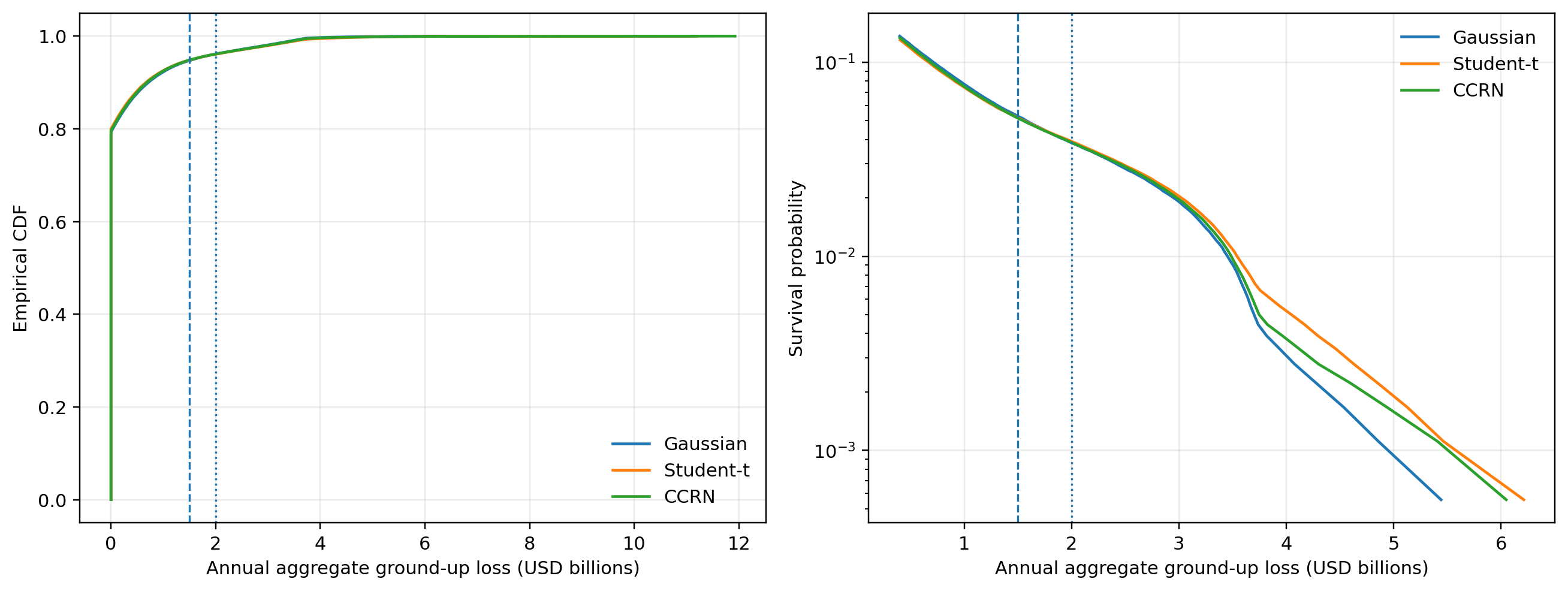}
\caption{Matched event-marginal comparison of annual aggregate ground-up loss.
The dashed line marks attachment and the dotted line marks exhaustion.}
\label{fig:loss-comparison-four-node}
\end{figure}

Figure~\ref{fig:loss-comparison-four-node} shows that the empirical
distributions nearly overlap around the attachment region but separate in the
upper tail. This explains why a low, narrow layer can have similar premiums
under alternative dependence models even when their retained-loss tails differ.

\subsection{Dependence Shape and Contract Sensitivity}
\label{subsec:dependence-contract-sensitivity}

The Student-$t$ degrees of freedom and the exogenous copula correlation are
varied independently to identify the portion of tail variation induced by the
benchmark specification.

\begin{table}[H]
\centering
\small
\caption{Student-$t$ degrees-of-freedom sensitivity. Tail quantities are in USD
billions.}
\label{tab:t-df-sensitivity}
\begin{tabular}{@{}ccccc@{}}
\toprule
$\boldsymbol{\nu}$ & \textbf{Attachment (\%)} & \textbf{Premium (\%)} & $\boldsymbol{\operatorname{VaR}_{99.5\%}}$ & $\boldsymbol{\operatorname{TVaR}_{99.5\%}}$ \\
\midrule
3  & 5.20 & 4.55 & 4.320 & 5.308 \\
5  & 5.17 & 4.49 & 4.049 & 5.016 \\
10 & 5.21 & 4.50 & 3.825 & 4.734 \\
20 & 5.16 & 4.42 & 3.734 & 4.567 \\
\bottomrule
\end{tabular}
\end{table}

Table~\ref{tab:t-df-sensitivity} shows how the Student-$t$ benchmark changes as
the degrees of freedom vary. Each row uses $350{,}000$ contract years and the
same underlying random-number stream. The $\nu=5$ row is the identical
Student-$t$ benchmark used in Table~\ref{tab:pricing-comparison}, so the
attachment probability and pure premium now coincide exactly across the two
tables up to displayed rounding. Lower degrees of freedom, which imply stronger
tail dependence, increase VaR and TVaR substantially, whereas attachment and
pure premium remain comparatively stable. The table therefore demonstrates that
the selected layer price and the far-tail capital metrics respond to different
regions of the annual loss distribution without introducing a separate-seed
comparison artifact.

\begin{table}[H]
\centering
\small
\caption{Equicorrelation sensitivity using positive-definite four-dimensional
correlation matrices. Tail quantities are in USD billions.}
\label{tab:rho-sensitivity}
\resizebox{\textwidth}{!}{%
\begin{tabular}{@{}lccccc@{}}
\toprule
\textbf{Family} & $\boldsymbol{\rho}$ & \textbf{Attachment (\%)} & \textbf{Premium (\%)} & $\boldsymbol{\operatorname{VaR}_{99.5\%}}$ & $\boldsymbol{\operatorname{TVaR}_{99.5\%}}$ \\
\midrule
Gaussian & 0.2 & 5.29 & 4.57 & 3.852 & 4.676 \\
Gaussian & 0.4 & 5.23 & 4.56 & 4.215 & 5.057 \\
Gaussian & 0.6 & 5.26 & 4.67 & 4.587 & 5.492 \\
Student-$t$ & 0.2 & 5.28 & 4.63 & 4.338 & 5.256 \\
Student-$t$ & 0.4 & 5.28 & 4.68 & 4.612 & 5.599 \\
Student-$t$ & 0.6 & 5.23 & 4.70 & 4.828 & 5.919 \\
\bottomrule
\end{tabular}}
\end{table}

Table~\ref{tab:rho-sensitivity} reports the effect of imposing common
correlations on the Gaussian and Student-$t$ benchmarks. Increasing
$\rho$ raises VaR and TVaR much more strongly than it changes attachment or
premium, and the Student-$t$ family retains the larger far-tail response at
each correlation level. The table isolates the sensitivity of the annual loss
distribution to dependence strength while holding the marginal construction
fixed.

\begin{figure}[H]
\centering
\includegraphics[width=0.94\textwidth]{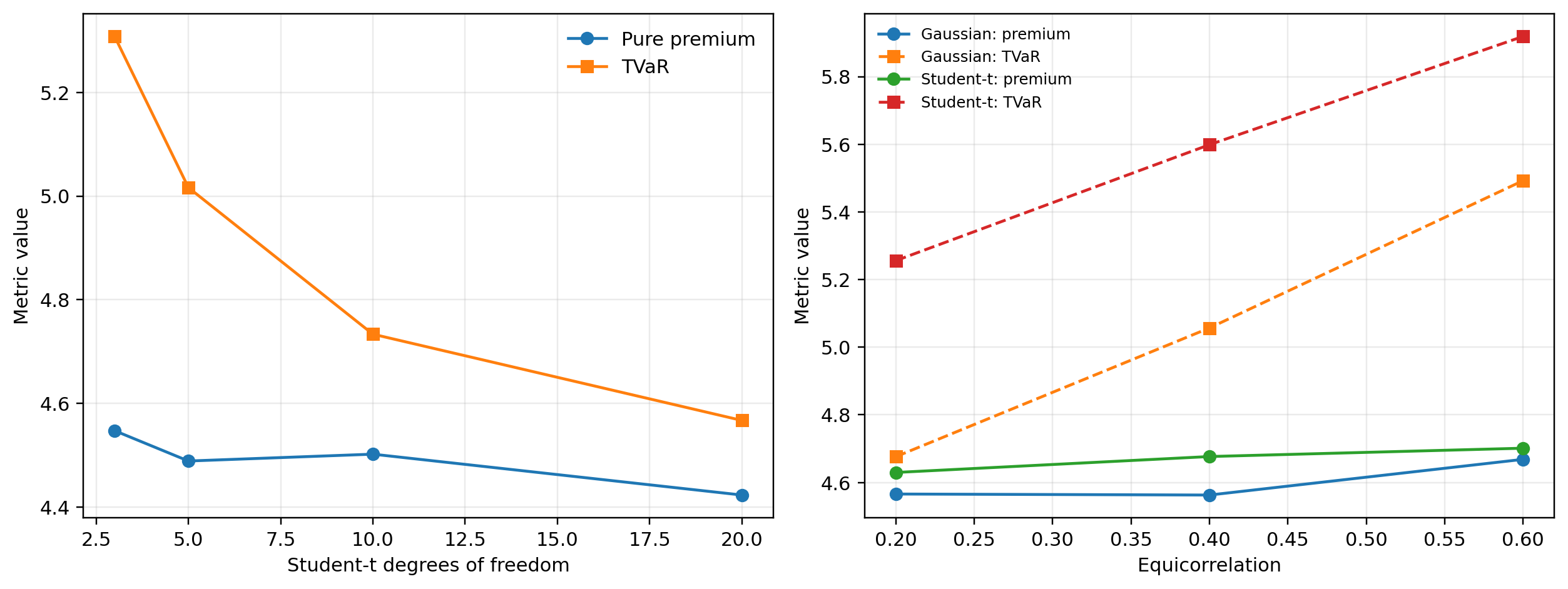}
\caption{Sensitivity of the copula benchmarks to Student-$t$ degrees of freedom
and exogenous equicorrelation.}
\label{fig:copula-sensitivity}
\end{figure}

Figure~\ref{fig:copula-sensitivity} visualizes the two benchmark sensitivity
experiments. The left panel shows that the Student-$t$ TVaR decreases markedly
as the degrees of freedom increase, while the premium changes only modestly.
The right panel shows that stronger equicorrelation produces a pronounced rise
in TVaR for both copula families, with a much smaller change in premium. The
figure makes the contrast between central-layer pricing and far-tail response
visible across the two dependence controls.

The economic effect of dependence changes with the contract position. Table~\ref{tab:contract-grid}
reports pure premiums across a grid of attachment points and limits using the
same annual loss samples.

\begin{table}[H]
\centering
\small
\caption{Pure premium as a percentage of layer limit across the
attachment--limit grid. Monetary contract values are in USD billions.}
\label{tab:contract-grid}
\begin{tabular}{@{}ccccc@{}}
\toprule
\textbf{Attachment} & \textbf{Limit} & \textbf{Gaussian} & \textbf{Student-$t$} & \textbf{CCRN} \\
\midrule
0.5 & 0.5 & 9.78 & 9.36 & 9.52 \\
0.5 & 1.0 & 8.06 & 7.76 & 7.86 \\
0.5 & 2.0 & 5.99 & 5.85 & 5.88 \\
1.5 & 0.5 & 4.49 & 4.49 & 4.46 \\
1.5 & 1.0 & 3.91 & 3.95 & 3.91 \\
1.5 & 2.0 & 2.90 & 2.99 & 2.93 \\
3.0 & 0.5 & 1.41 & 1.57 & 1.48 \\
3.0 & 1.0 & 0.96 & 1.15 & 1.03 \\
3.0 & 2.0 & 0.57 & 0.74 & 0.64 \\
\bottomrule
\end{tabular}
\end{table}

Table~\ref{tab:contract-grid} shows how the normalized pure premium changes as
the attachment point and layer limit are varied. Premiums decline as the layer
moves farther into the tail, and the differences among the models become more
visible at the highest attachment level. At the reference attachment and limit,
the three premiums remain close, whereas high layers place greater weight on
the Student-$t$ tail and therefore exhibit a clearer model separation.

\begin{figure}[H]
\centering
\includegraphics[width=0.94\textwidth]{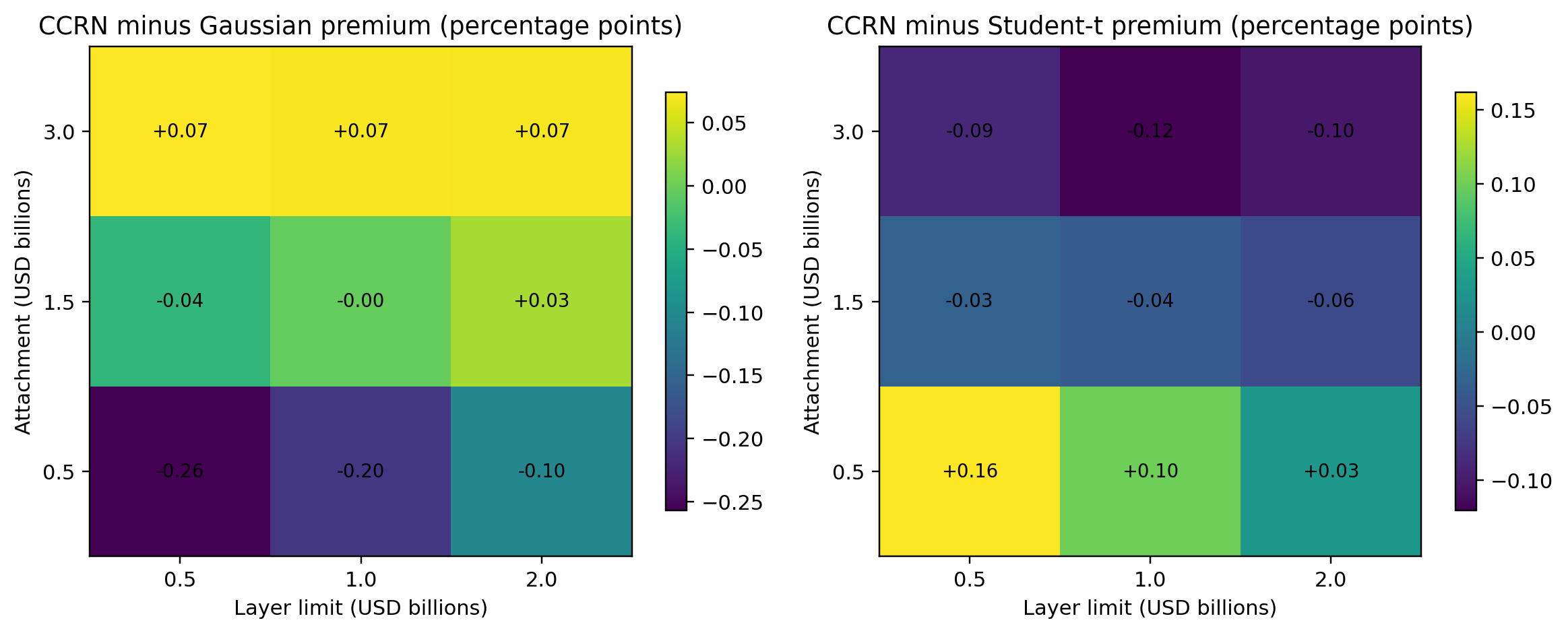}
\caption{Difference between the CCRN premium and each matched-marginal copula
premium across the attachment--limit grid.}
\label{fig:contract-sensitivity}
\end{figure}

Figure~\ref{fig:contract-sensitivity} displays the signed premium differences
between the CCRN and each copula benchmark over the attachment--limit grid. The
left heatmap shows that the CCRN premium moves from below the Gaussian premium
at low attachment to slightly above it at high attachment. The right heatmap
shows that the CCRN exceeds the Student-$t$ premium for the lowest layers but
falls below it for high-attachment layers. The figure therefore converts the
tail-distribution differences into a direct actuarial measure of where model
choice affects contract price.

\subsection{Climate-Conditioned Frequency, Latent Dependence, and the Upper Corner}
\label{subsec:sensitivity-bound}

Climate enters the edge hazards through Equation~\eqref{eq:dynamic-edge-weight}.
The baseline frequency model holds the Poisson mean fixed. A second design uses
\[
N_y\mid\Lambda_y\sim\operatorname{Poisson}(\Lambda_y),\qquad
\log\Lambda_y=\log(1.50)+0.12\,\Delta T,
\]
so the climate input affects both propagation and annual event frequency.

\begin{table}[H]
\centering
\small
\caption{Climate-input sensitivity under fixed and climate-conditioned annual
frequency. Tail quantities are in USD billions.}
\label{tab:climate-frequency-sensitivity}
\resizebox{\textwidth}{!}{%
\begin{tabular}{@{}clcccc@{}}
\toprule
$\boldsymbol{\Delta T}$ & \textbf{Frequency model} & \textbf{Attachment (\%)} & \textbf{Premium (\%)} & $\boldsymbol{\operatorname{VaR}_{99.5\%}}$ & $\boldsymbol{\operatorname{TVaR}_{99.5\%}}$ \\
\midrule
0.0 & Fixed Poisson & 3.93 & 3.31 & 3.632 & 4.377 \\
0.0 & Climate-conditioned Poisson & 3.93 & 3.31 & 3.632 & 4.377 \\
1.0 & Fixed Poisson & 4.93 & 4.25 & 3.745 & 4.819 \\
1.0 & Climate-conditioned Poisson & 5.56 & 4.83 & 3.847 & 4.894 \\
2.0 & Fixed Poisson & 6.25 & 5.47 & 4.117 & 5.269 \\
2.0 & Climate-conditioned Poisson & 7.84 & 6.85 & 4.567 & 5.669 \\
\bottomrule
\end{tabular}}
\end{table}

Table~\ref{tab:climate-frequency-sensitivity} compares two channels through
which the temperature-anomaly input can affect annual losses. Under fixed
frequency, the climate input changes only the propagation hazards; under the
climate-conditioned Poisson specification, it also changes the deterministic
Poisson rate according to $\Lambda_y=1.50\exp(0.12\,\Delta T)$. The two
designs use common random numbers and therefore coincide exactly at
$\Delta T=0$, where both reduce to $N_y\sim\operatorname{Poisson}(1.50)$.
For positive anomalies, attachment, premium, VaR, and TVaR rise under both
designs, with a stronger increase when the annual-frequency channel is active.

\begin{figure}[H]
\centering
\includegraphics[width=0.82\textwidth]{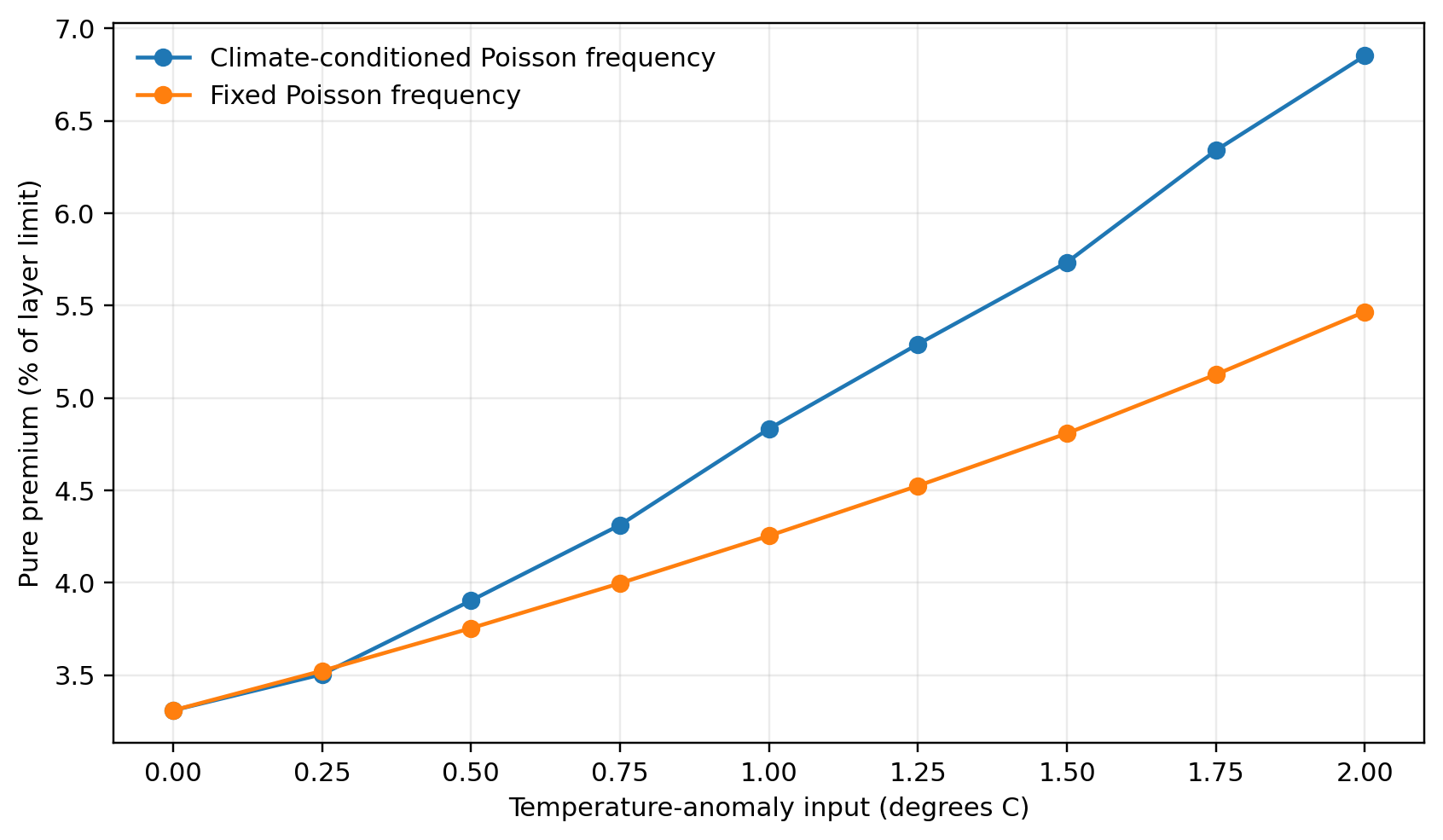}
\caption{Pure-premium sensitivity when the climate input affects propagation
alone and when it also affects annual event frequency.}
\label{fig:premium-sensitivity-four-node}
\end{figure}

Figure~\ref{fig:premium-sensitivity-four-node} traces the pure premium over the
full temperature-anomaly grid for the two frequency specifications. The curves
are close near the reference state but diverge as the anomaly increases,
because the climate-conditioned Poisson design amplifies both event
propagation and annual event count. The figure therefore separates the premium
response attributable to propagation alone from the additional frequency
effect within the same contract map.

Residual dependence is introduced only between the exogenous root-occurrence
uniforms for fuel aridity and precipitation. Downstream wildfire and debris-flow
onsets continue to be generated by the directed cascade, preventing duplication
of the encoded propagation channels.

\begin{table}[H]
\centering
\small
\caption{Sensitivity to latent correlation between exogenous root-occurrence
uniforms. Tail quantities are in USD billions.}
\label{tab:latent-dependence}
\begin{tabular}{@{}ccccc@{}}
\toprule
$\boldsymbol{\rho_U}$ & \textbf{Attachment (\%)} & \textbf{Premium (\%)} & $\boldsymbol{\operatorname{VaR}_{99.5\%}}$ & $\boldsymbol{\operatorname{TVaR}_{99.5\%}}$ \\
\midrule
0.0 & 5.09 & 4.41 & 3.776 & 4.807 \\
0.2 & 5.11 & 4.44 & 3.918 & 4.999 \\
0.4 & 5.14 & 4.48 & 4.138 & 5.214 \\
\bottomrule
\end{tabular}
\end{table}

Table~\ref{tab:latent-dependence} shows the effect of residual dependence
between the exogenous fuel-aridity and precipitation onset uniforms. As
$\rho_U$ increases, all reported risk measures rise, but the changes in VaR and
TVaR are more pronounced than the changes in attachment and premium. The table
indicates that unobserved common root conditions affect the far tail more
strongly than the selected contract attachment region.

\begin{figure}[H]
\centering
\includegraphics[width=0.76\textwidth]{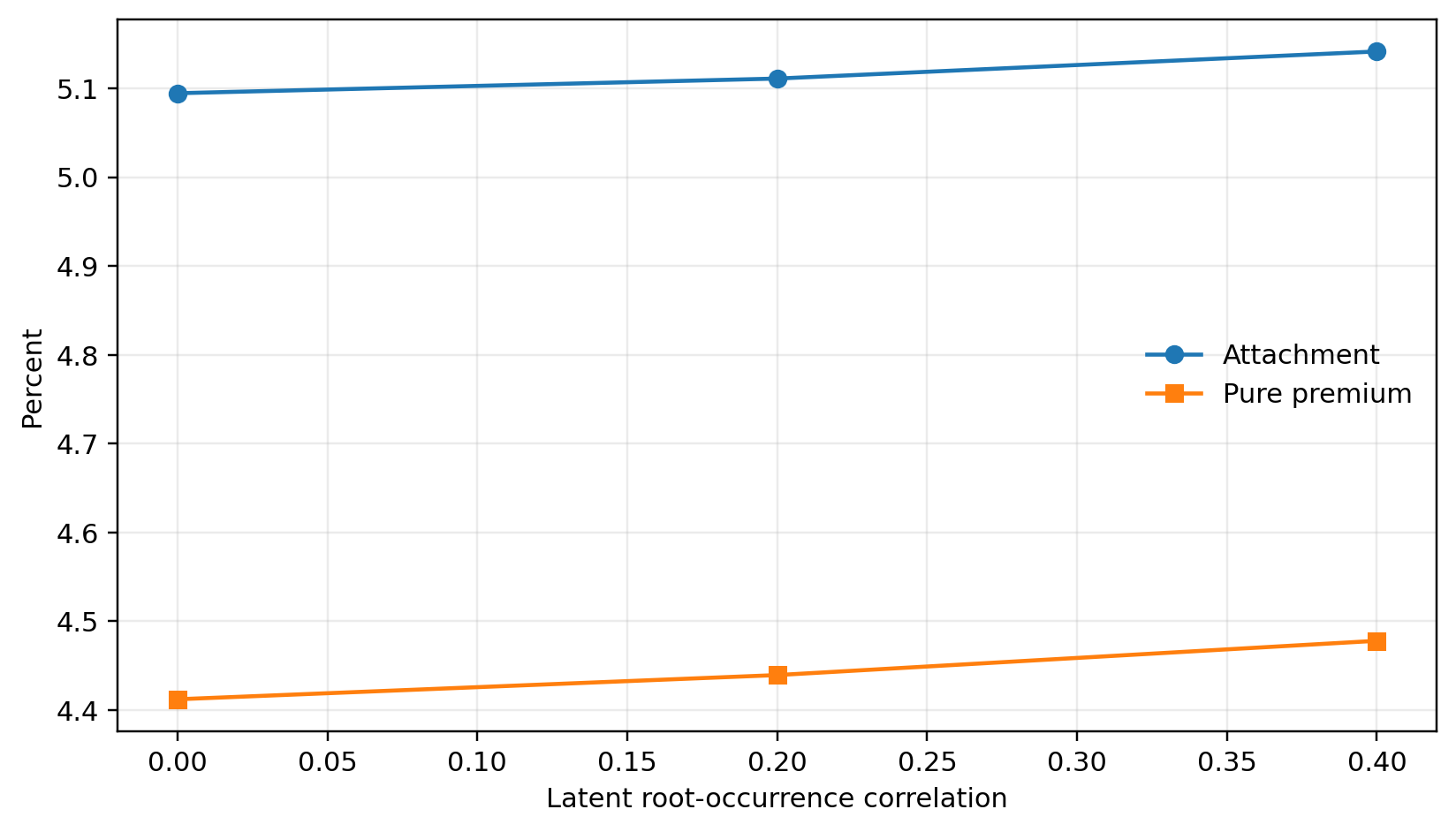}
\caption{Contract-level response to residual dependence among exogenous root
onsets.}
\label{fig:latent-dependence}
\end{figure}

Figure~\ref{fig:latent-dependence} visualizes the contract-level response to
increasing residual root dependence. The premium and attachment curves rise
only gradually, whereas the tail metrics exhibit a steeper increase. This
pattern is consistent with correlated exogenous conditions creating more
opportunities for joint wildfire--precipitation states without replacing the
directional dependence already represented by the cascade.

The pathwise upper-corner check uses $30{,}000$ paired event-window scenarios.
For each pair, the exogenous shocks and onset uniforms are held fixed, while the
edge probabilities, wind multiplier, and terrain multiplier vary within the
documented rectangular stress set. All paired losses satisfy the analytical
ordering. The upper corner produces a strict increase in $1{,}683$ scenarios,
and the smallest positive slack is $3.59\times10^{-5}$ billion USD.

\begin{figure}[H]
\centering
\includegraphics[width=0.76\textwidth]{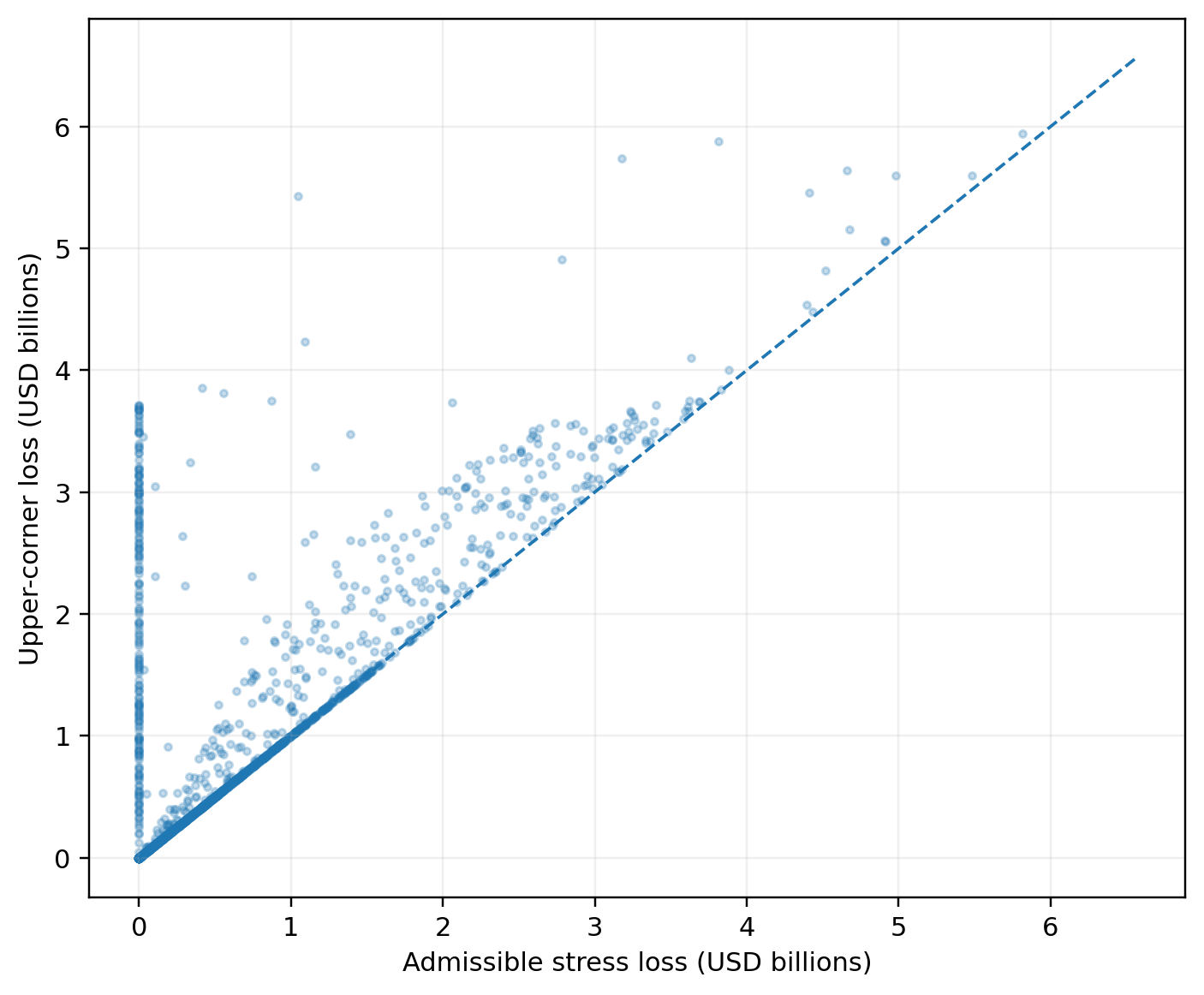}
\caption{Pathwise upper-corner verification under common aleatory inputs.}
\label{fig:pathwise-bound-verification}
\end{figure}

Figure~\ref{fig:pathwise-bound-verification} plots each admissible-stress loss
against the loss obtained from the corresponding upper-corner scenario under
the same aleatory inputs. Every point lies on or above the diagonal reference
line, so the upper-corner loss is never smaller than the paired admissible loss.
Points on the diagonal represent scenarios for which the stress increase does
not change the realized loss, while points above it show strict pathwise
amplification. The figure provides a direct numerical verification of the
ordering used in Theorem~\ref{thm:certified-bound}.

\subsection{Structural Ablation and One-at-a-Time Parameter Sensitivity}
\label{robustness}

The ablation study removes one structural component at a time while retaining
common random-number seeds.

\begin{table}[H]
\centering
\small
\caption{Ablation results at $\Delta T=1.2^\circ\mathrm{C}$. Tail quantities
are annual aggregate ground-up losses in USD billions.}
\label{tab:ablation}
\resizebox{\textwidth}{!}{%
\begin{tabular}{@{}lcccc@{}}
\toprule
\textbf{Specification} & \textbf{Attachment (\%)} & \textbf{Premium (\%)} & $\boldsymbol{\operatorname{VaR}_{99.5\%}}$ & $\boldsymbol{\operatorname{TVaR}_{99.5\%}}$ \\
\midrule
Full CCRN & 5.16 & 4.48 & 3.767 & 4.828 \\
No directional propagation & 0.87 & 0.51 & 1.738 & 2.181 \\
No debris-flow interaction & 5.12 & 4.43 & 3.699 & 4.435 \\
No demand surge & 5.01 & 4.33 & 3.720 & 4.681 \\
Single-event approximation & 3.35 & 2.88 & 3.503 & 4.002 \\
\bottomrule
\end{tabular}}
\end{table}

Table~\ref{tab:ablation} quantifies the contribution of each structural
component by comparing the full CCRN with one-component removals. Eliminating
directional propagation produces the largest reduction in attachment,
premium, and tail risk under the reference parameterization. Replacing annual
aggregation with a single-event approximation also reduces the contract
metrics materially, while removing demand surge or the debris-flow interaction
has a smaller effect near attachment. The interaction nevertheless has a more
visible effect on TVaR than on the selected layer price.

\begin{figure}[H]
\centering
\includegraphics[width=0.88\textwidth]{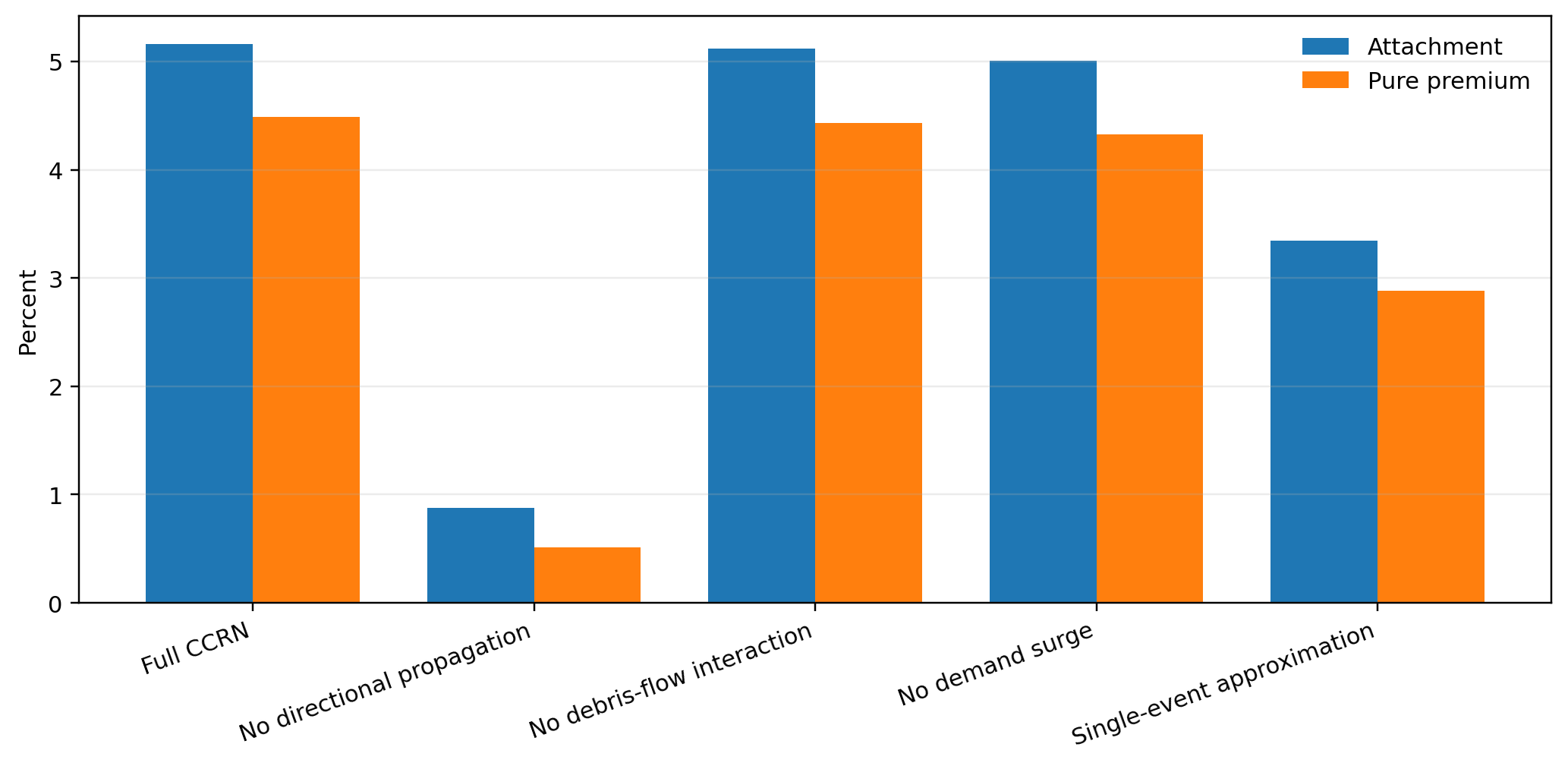}
\caption{Attachment probability and pure premium under structural ablations.}
\label{fig:ablation}
\end{figure}

Figure~\ref{fig:ablation} presents the attachment probability and pure premium
for the same ablation specifications in a directly comparable graphical form.
The sharp decline produced by removing directional propagation is visible in
both contract metrics, followed by the reduction under the single-event
approximation. The smaller gaps for demand surge and the debris-flow interaction
show that their influence is more localized under the selected layer and
reference parameter vector.

A structured low--base--high one-at-a-time parameter sweep then varies the baseline edge
probabilities, climate coefficients, background occurrence probabilities,
interaction strength, severity-response steepness, annual frequency, and
demand-surge intensity. Table~\ref{tab:oat-sensitivity} reports the percentage
change in pure premium relative to the base specification.

\begin{table}[H]
\centering
\small
\caption{One-at-a-time sensitivity of the pure premium relative to the base
specification.}
\label{tab:oat-sensitivity}
\begin{tabular}{@{}lrr@{}}
\toprule
\textbf{Parameter} & \textbf{Low setting (\%)} & \textbf{High setting (\%)} \\
\midrule
Fuel-to-wildfire edge probability & -25.0 & 26.6 \\
Debris interaction probability & -0.4 & 0.4 \\
Fuel-edge climate coefficient & -11.9 & 14.1 \\
Debris-edge climate coefficient & -0.1 & 0.2 \\
Wildfire background probability & -3.5 & 6.8 \\
Debris background probability & -0.0 & 0.1 \\
Interaction hazard multiplier & -0.2 & 0.3 \\
Severity-response steepness & -13.5 & 7.3 \\
Annual event frequency & -21.0 & 20.7 \\
Demand-surge intensity & -1.8 & 1.7 \\
\bottomrule
\end{tabular}
\end{table}

Table~\ref{tab:oat-sensitivity} reports the percentage change in pure premium
when each input is moved separately from its base value to a low or high
setting. The fuel-to-wildfire edge probability and annual event frequency
produce the largest two-sided changes, followed by the fuel-edge climate
coefficient and severity-response steepness. The debris-flow background and
interaction parameters have comparatively small effects on the reference
layer premium, indicating that their main influence occurs farther in the loss
tail.

\begin{figure}[H]
\centering
\includegraphics[width=0.88\textwidth]{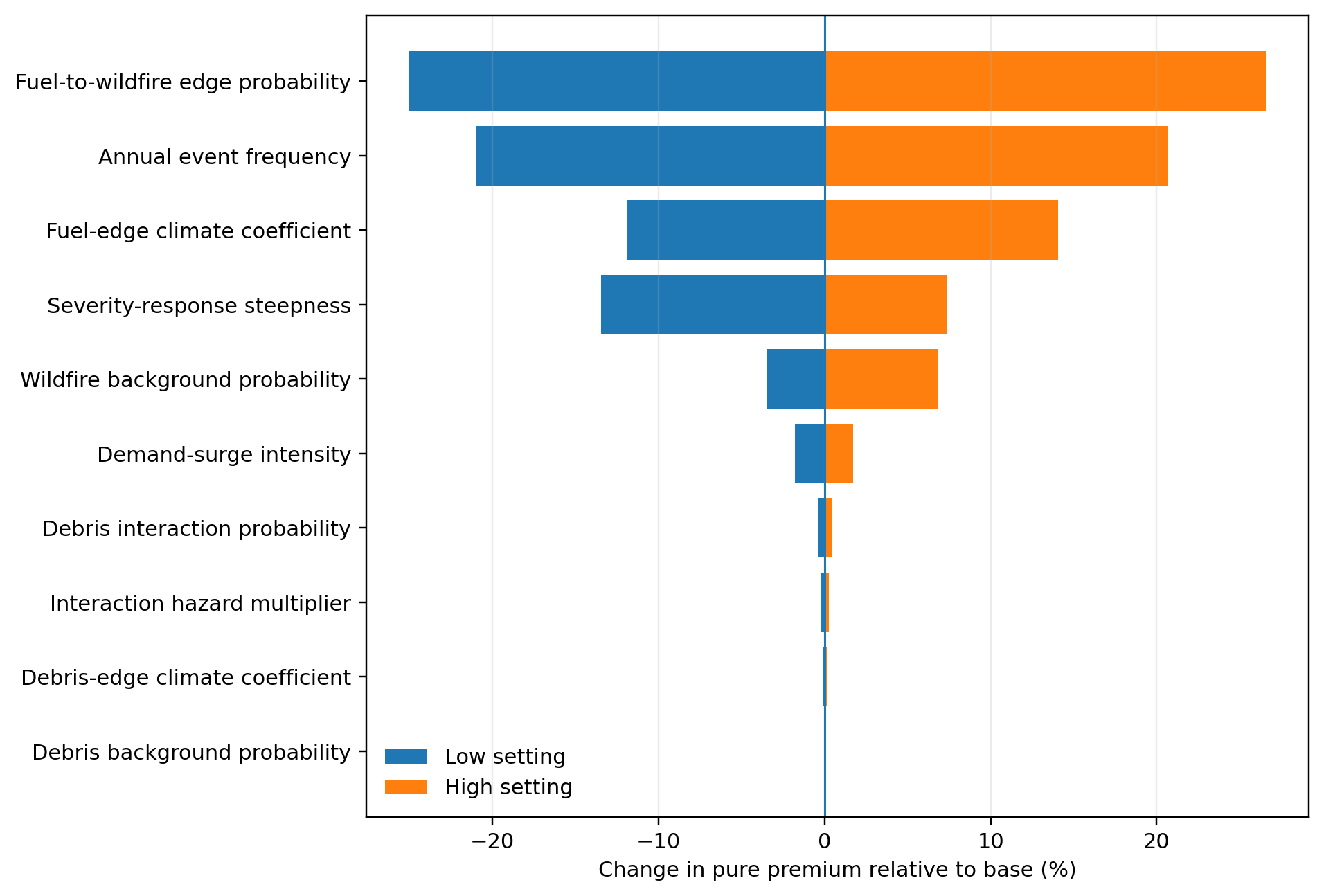}
\caption{Tornado chart for the structured low--base--high one-at-a-time pure-premium sensitivity sweep.}
\label{fig:oat-sensitivity}
\end{figure}

Figure~\ref{fig:oat-sensitivity} ranks the low--high premium changes from
Table~\ref{tab:oat-sensitivity} in tornado form. The longest bars correspond to
the fuel-to-wildfire edge probability and annual event frequency, with the
fuel-edge climate coefficient and severity-response steepness forming the next
tier. The compact bars for the debris-flow interaction and background terms
show that these inputs are less influential for the chosen layer, even though
they can remain relevant to far-tail loss.

\subsection{Parameter Recovery, Predictive Uncertainty, and a Static DAG Benchmark}
\label{subsec:uncertainty-bn}

A synthetic recovery experiment uses $400{,}000$ event windows with
$\Delta T\sim\operatorname{Uniform}(0,2)$ and re-estimates the two edge
models using complementary-log-log regressions. The fuel and debris stress
indices enter as offsets, so the fitted intercepts recover the reference
cumulative hazards and the slope coefficients recover the climate sensitivities.
For each edge, the fitted intercept--slope pair is sampled jointly from its
bivariate asymptotic normal approximation, thereby preserving the estimated
within-edge covariance. The two edge blocks are sampled independently of one
another, and no cross-block covariance is imposed. Background occurrence
probabilities use Jeffreys-beta posterior draws based on the same $400{,}000$
event windows. Annual frequency uses a Jeffreys-gamma posterior constructed from
a separate synthetic sample of $30{,}000$ contract years. Loss-capacity and
demand-surge multipliers are sampled independently from mean-one lognormal laws
with log-scale standard deviations $0.10$ and $0.15$, respectively.

The uncertainty-propagation experiment retains $B_{\theta}=100$ complete
parameter vectors. Each vector is propagated through $25{,}000$ newly simulated
contract years. The Monte Carlo-only benchmark consists of $100$ independent
replications of $25{,}000$ contract years at the fixed reference parameter
vector. The aleatory random-number streams used for the parameter-propagation
and fixed-parameter rows are independent. Thus, the reported comparison reflects
the distribution of replication-level contract statistics rather than a paired
common-random-number contrast.

\begin{table}[H]
\centering
\small
\caption{Reference values and central 90\% intervals used in the
predictive uncertainty propagation.}
\label{tab:parameter-recovery}
\begin{tabular}{@{}lrrrr@{}}
\toprule
\textbf{Parameter} & \textbf{Reference} & \textbf{Median} & \textbf{Lower} & \textbf{Upper} \\
\midrule
$p_{12}^{0}$ & 0.3500 & 0.3508 & 0.3440 & 0.3594 \\
$\beta_{12}$ & 0.3500 & 0.3554 & 0.3269 & 0.3731 \\
$p_{(2,3)4}^{0}$ & 0.5500 & 0.5401 & 0.5003 & 0.5822 \\
$\beta_{(2,3)4}$ & 0.1800 & 0.1879 & 0.0869 & 0.2854 \\
$q_2$ & 0.0150 & 0.0150 & 0.0147 & 0.0154 \\
$q_4$ & 0.0010 & 0.0010 & 0.0009 & 0.0011 \\
Annual frequency & 1.5000 & 1.4982 & 1.4880 & 1.5119 \\
Loss-capacity multiplier & 1.0000 & 0.9892 & 0.8505 & 1.1363 \\
Demand-surge multiplier & 1.0000 & 0.9893 & 0.7808 & 1.2302 \\
\bottomrule
\end{tabular}
\end{table}

Table~\ref{tab:parameter-recovery} compares the reference inputs with the
medians and central 90\% intervals obtained in the synthetic recovery experiment.
For each parameter, the interval endpoints are the empirical 5th and 95th
percentiles of the $100$ retained parameter vectors. The baseline edge
probabilities, background occurrence rates, and annual frequency remain close to
their generating values. Wider intervals appear for the debris-flow climate
coefficient and the financial-map multipliers, reflecting the weaker information
available for a rare interaction and the additional uncertainty assigned to loss
capacity and demand surge. These are marginal summaries of the blockwise
parameter-draw construction described above; they should not be interpreted as
independent rectangular uncertainty bounds.

\begin{table}[H]
\centering
\small
\caption{Central 90\% ranges for Monte Carlo process variation and
combined parameter--modeled-process uncertainty.}
\label{tab:uncertainty-propagation}
\resizebox{\textwidth}{!}{%
\begin{tabular}{@{}llccc@{}}
\toprule
\textbf{Metric} & \textbf{Uncertainty source} & \textbf{Lower} & \textbf{Median} & \textbf{Upper} \\
\midrule
Attachment (\%) & Monte Carlo only & 4.936 & 5.110 & 5.369 \\
Attachment (\%) & Parameter + modeled process & 4.272 & 5.090 & 5.920 \\
Premium (\%) & Monte Carlo only & 4.242 & 4.403 & 4.638 \\
Premium (\%) & Parameter + modeled process & 3.653 & 4.385 & 5.165 \\
$\operatorname{VaR}_{99.5\%}$ & Monte Carlo only & 3.703 & 3.764 & 3.906 \\
$\operatorname{VaR}_{99.5\%}$ & Parameter + modeled process & 3.180 & 3.733 & 4.281 \\
$\operatorname{TVaR}_{99.5\%}$ & Monte Carlo only & 4.586 & 4.820 & 5.032 \\
$\operatorname{TVaR}_{99.5\%}$ & Parameter + modeled process & 4.017 & 4.726 & 5.589 \\
\bottomrule
\end{tabular}}
\end{table}

Table~\ref{tab:uncertainty-propagation} separates simulation-process variation
from the combined effect of the specified parameter uncertainty and modeled
aleatory processes. For each metric and each row, the displayed central 90\%
range is the empirical 5th--95th percentile interval across the corresponding
$100$ replication-level statistic values; annual losses are not pooled across
replications before the interval is formed. The central medians remain close
across the two sources, while the combined intervals are substantially wider.

Conditional beta innovations governed by $\phi_j$ in
Equation~\eqref{eq:conditional-severity-model} are used in estimation and
diagnostic calibration but are not propagated as an additional aleatory
severity term in this controlled forward experiment. The cascade continues to
propagate the fitted conditional mean $\Psi_j$ deterministically within each
realization, consistent with Section~\ref{subsec:estimation-aggregation}. The
combined ranges therefore include uncertainty in the occurrence/climate-response
parameters, background rates, annual frequency, and financial-map multipliers,
together with Monte Carlo process variation, but exclude residual conditional-
severity variation around the fitted mean response. They are consequently a
scoped predictive-uncertainty assessment rather than an interval encompassing
all possible severity uncertainty.

\begin{figure}[H]
\centering
\includegraphics[width=0.90\textwidth]{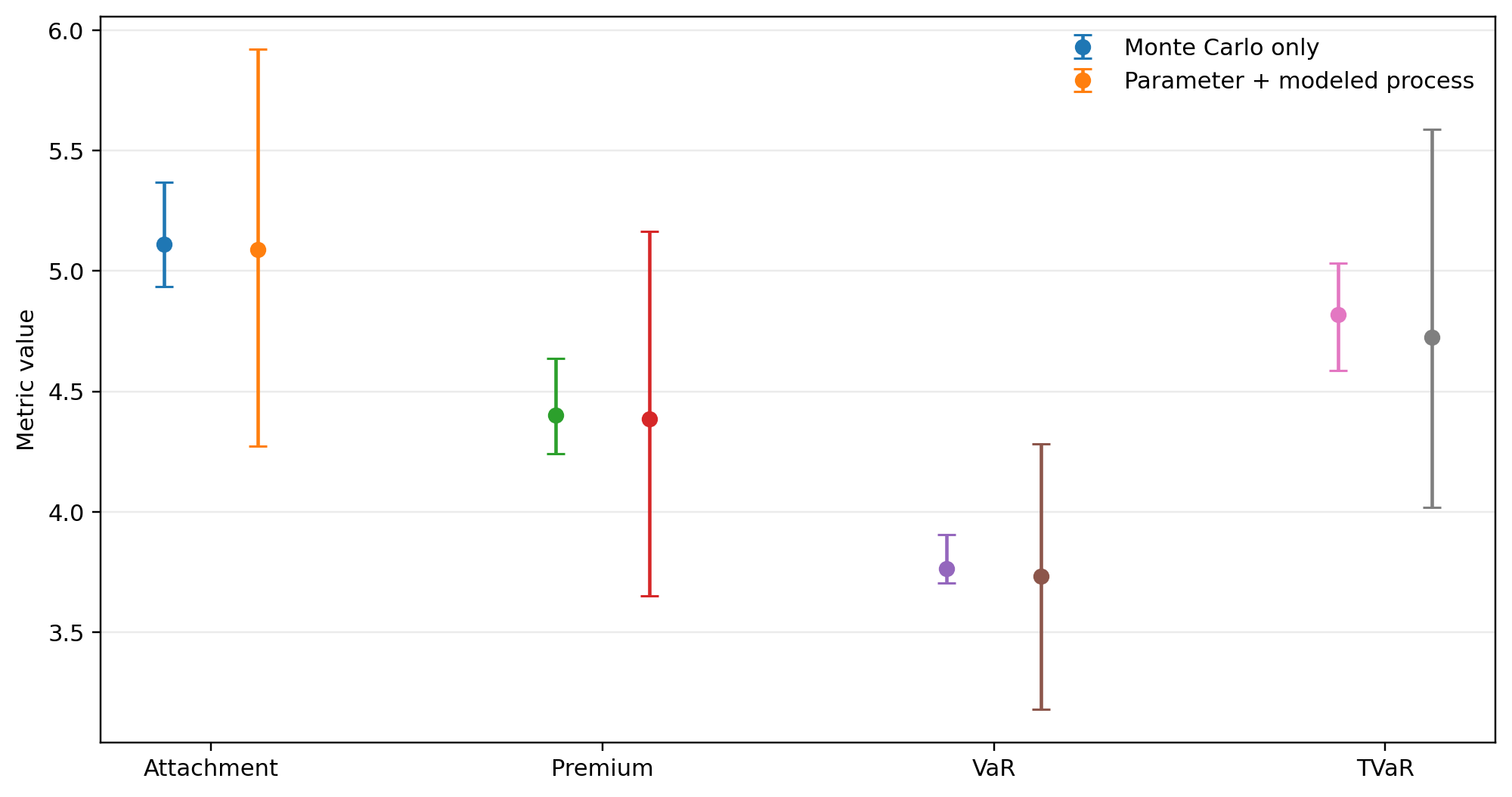}
\caption{Monte Carlo process variation compared with combined parameter and
modeled-process uncertainty; residual conditional-severity innovations are excluded.}
\label{fig:uncertainty-propagation}
\end{figure}

Figure~\ref{fig:uncertainty-propagation} displays the interval expansion from
Monte Carlo variation to combined parameter--modeled-process uncertainty for each
reported metric. The broader combined ranges are visible for both the
contract-level quantities and the far-tail measures, while their centers remain
similar. The figure demonstrates why a small simulation standard error cannot
be used as a substitute for predictive uncertainty in model inputs and the
financial mapping.

A structurally closer benchmark is fitted on an independent CCRN training sample
of $500{,}000$ event windows and evaluated on $250{,}000$ independently simulated
contract years per model. The graph, annual frequency law, bounded financial map,
and contract are held fixed. No continuous parent variable is discretized.
Instead, the static benchmark deliberately integrates out the continuous parent
severity and modifier states when constructing Bernoulli conditional-probability
tables from the discrete parent-occurrence states: wildfire occurrence is
conditioned on the fuel-aridity occurrence indicator, and debris-flow occurrence
is conditioned on the joint wildfire and precipitation occurrence indicators.
All Bernoulli table entries use Jeffreys additive smoothing with pseudocount
$0.5$ for each child state.

Conditional severity is represented by empirical positive-severity pools frozen
from the training sample. Root-node severities are drawn from their positive
marginal pools; wildfire severity is resampled within the corresponding
fuel-aridity parent state; and debris-flow severity is resampled within the
corresponding wildfire--precipitation parent state. To prevent unstable sampling
from very sparse cells, a wildfire pool containing fewer than $50$ positive
observations or a debris-flow pool containing fewer than $30$ is replaced by the
corresponding positive child-severity marginal pool from the training sample.
Every conditional probability, smoothing constant, pool, and fallback decision is
fixed before the test simulation and is not re-estimated on the test sample.

\begin{table}[H]
\centering
\small
\caption{Out-of-sample comparison with a static Bayesian network on the same
graph and financial mapping. Tail quantities are in USD billions. Parentheses
contain model-wise bootstrap standard errors from 60 i.i.d. contract-year
resamples, conditional on the frozen fitted benchmark.}
\label{tab:static-bn-comparison}
\begin{tabular}{@{}lcccc@{}}
\toprule
\textbf{Model} & \textbf{Attachment (\%)} & \textbf{Premium (\%)} & $\boldsymbol{\operatorname{VaR}_{99.5\%}}$ & $\boldsymbol{\operatorname{TVaR}_{99.5\%}}$ \\
\midrule
CCRN test sample & 5.18 (0.04) & 4.49 (0.04) & 3.761 (0.012) & 4.836 (0.033) \\
Static Bayesian network & 5.18 (0.05) & 4.47 (0.05) & 3.743 (0.009) & 4.688 (0.036) \\
\bottomrule
\end{tabular}
\end{table}

Table~\ref{tab:static-bn-comparison} compares the CCRN with the frozen static
Bayesian-network benchmark on independent held-out simulations. Attachment and
pure premium are nearly identical relative to their model-wise Monte Carlo
precision, while the static benchmark has lower VaR and, more clearly, lower
TVaR. The TVaR difference is $0.149$ billion USD. Combining the two independent
model-wise bootstrap standard errors gives a Monte Carlo standard error of
approximately $0.048$ billion for that difference, so the observed TVaR gap is
about $3.1$ Monte Carlo standard errors. This comparison is conditional on the
frozen training-sample fit and therefore does not include re-estimation
uncertainty for the static network.

The comparison provides a closer assessment of retaining continuous
state-dependent propagation relative to a benchmark that uses the same graph,
annual-frequency and financial mappings but represents conditional relationships
through static discrete-parent probability tables and empirical severity pools.
It does not isolate propagation alone, because marginalization of the continuous
parent state and the static conditional representation change simultaneously.

\subsection{Independent Monte Carlo Precision}
\label{subsec:mc-precision}

Independent replications are used to assess Monte Carlo precision at increasing
contract-year sample sizes. Table~\ref{tab:independent-mc} reports the
across-replication mean and standard deviation from ten replications at each
sample size.

\begin{table}[H]
\centering
\small
\caption{Independent-replication Monte Carlo precision. Each entry is mean
(standard deviation) across ten replications.}
\label{tab:independent-mc}
\resizebox{\textwidth}{!}{%
\begin{tabular}{@{}rcccc@{}}
\toprule
\textbf{Contract years} & \textbf{Attachment (\%)} & \textbf{Premium (\%)} & $\boldsymbol{\operatorname{VaR}_{99.5\%}}$ & $\boldsymbol{\operatorname{TVaR}_{99.5\%}}$ \\
\midrule
25,000 & 5.122 (0.091) & 4.432 (0.091) & 3.790 (0.050) & 4.813 (0.108) \\
50,000 & 5.119 (0.098) & 4.418 (0.089) & 3.785 (0.045) & 4.794 (0.059) \\
100,000 & 5.138 (0.078) & 4.424 (0.076) & 3.756 (0.015) & 4.810 (0.056) \\
200,000 & 5.133 (0.066) & 4.431 (0.051) & 3.763 (0.014) & 4.813 (0.051) \\
\bottomrule
\end{tabular}}
\end{table}

Table~\ref{tab:independent-mc} reports the across-replication mean and standard
deviation of each statistic at increasing simulation budgets. The mean values
remain stable across sample sizes, while the standard deviations generally
narrow as the number of contract years increases. Attachment and premium are
estimated more precisely than the extreme-tail quantities, particularly TVaR,
which requires more simulated years for comparable replication stability.

\begin{figure}[H]
\centering
\includegraphics[width=0.78\textwidth]{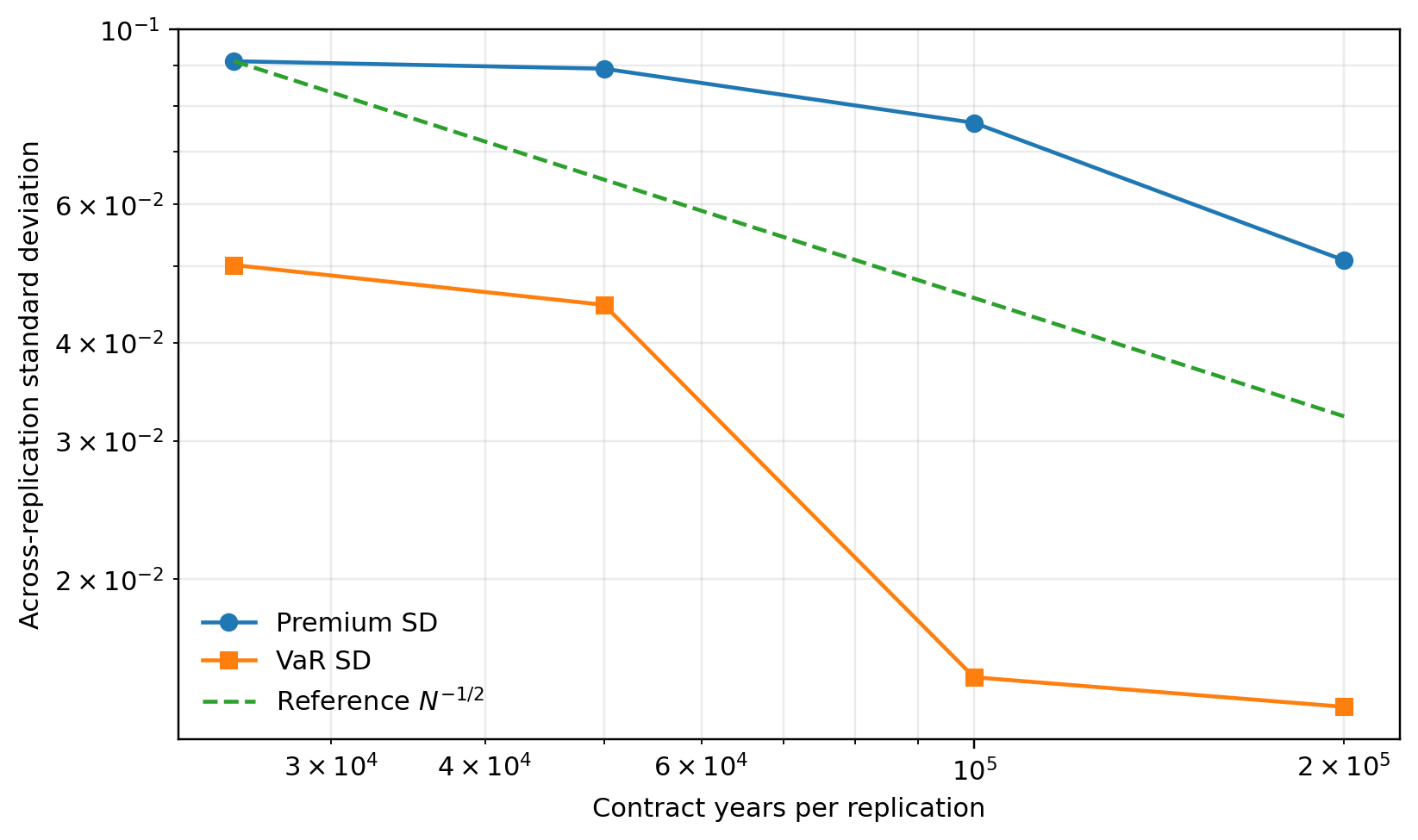}
\caption{Across-replication standard deviations as the contract-year sample
increases. The dashed line gives an $N^{-1/2}$ reference slope.}
\label{fig:mc-convergence}
\end{figure}

Figure~\ref{fig:mc-convergence} plots the across-replication standard
deviations from Table~\ref{tab:independent-mc} against the contract-year sample
size. The downward trajectories are broadly consistent with the displayed
$N^{-1/2}$ reference rate, although the tail estimators remain more variable
than attachment and premium. Together with the model-wise and bootstrap
standard errors in Tables~\ref{tab:pricing-comparison}
and~\ref{tab:tail-comparison}, the figure documents the numerical precision of
the reported comparisons using independent replications.

\section{Computational Complexity and Numerical Verification}
\label{analy}

Let $n=|\mathcal{V}|$, $m=|\mathcal{E}|$, $N_{\mathrm{yr}}$ be the number of simulated contract years, and $N_{\mathrm{evt}}=\sum_{y=1}^{N_{\mathrm{yr}}}N_y$ the realized number of event windows. A topological ordering is computed once in $O(n+m)$ time. Each event-window cascade then evaluates every node and edge once, so the propagation and bounded loss mapping require
\begin{equation}
O\!\left(N_{\mathrm{evt}}(n+m)\right)
\label{eq:cascade-complexity}
\end{equation}
operations. Annual aggregation and application of the layer require an additional $O(N_{\mathrm{evt}}+N_{\mathrm{yr}})$ operations. For a sparse graph with $m=O(n)$ and finite mean event frequency, expected simulation cost is $O(N_{\mathrm{yr}}n)$.

A fully vectorized implementation stores $O(N_{\mathrm{evt}}n+m)$ working values. Batching event windows reduces working memory to $O(Bn+m)$ for batch size $B$ without changing the model. The upper-corner stress evaluation adds one coupled cascade evaluation per aleatory scenario and does not require linear programming for a rectangular monotone set.

These statements concern forward simulation, not calibration. Hazard estimation, marginal fitting, topology assessment, and dynamic-copula or point-process estimation must be timed separately. Controlled comparisons should use common hardware, numerical precision, random-number budgets, forecast origins, and stopping rules.

The numerical validation additionally checks state bounds, component loss caps,
climate monotonicity under common random numbers, annual aggregation, the
upper-corner inequality, and equality of the fixed- and climate-conditioned
frequency specifications at $\Delta T=0$.

\section{Conclusion}
\label{con}

This paper introduced a climate-conditioned cascade network (CCRN) that decouples calendar-scale climate covariates, event-window propagation, and capacity-bounded financial mapping for multi-peril reinsurance. The framework guarantees a unique finite-step closure on a DAG and provides a strict pathwise upper-corner loss bound under monotone stress, a result verified exactly by paired numerical experiments.

Key numerical findings reveal that the choice of dependence model (such as Gaussian, $t$-Student, or CCRN) primarily impacts far-tail and high-layer losses, while central prices remain stable. Structured sensitivity analyses identify directional fuel-to-wildfire propagation and annual event frequency as the dominant drivers of high-layer losses. Furthermore, uncertainty in occurrence and climate-response parameters generates substantially wider predictive ranges than inherent Monte Carlo process variation, highlighting the importance of robust parameter estimation.

Compared to static Bayesian networks, the CCRN retains critical state-dependent tail structures that standard models tend to smooth out. This framework offers a rigorous foundation for empirical estimation using timestamped hazard and claims data, facilitating direct application to fitted multi-peril portfolios.

\section*{Acknowledgements}
The authors used AI solely for language refinement and LaTeX formatting. The authors reviewed and verified the final manuscript and take full responsibility for its content.

\appendix
\section{Detailed Data Construction and Estimation Workflow}
\label{app:data-workflow}

The material in this appendix operationalizes the empirical protocol summarized
in Section~\ref{data} and provides the detailed construction underlying the
synthetic recovery and uncertainty-propagation experiment.

\subsection{Mechanistic Specification and Data Construction}
\label{subsec:mechanistic-data}

An empirical implementation begins by fixing the physical structure and the rules used to construct observable event states. The event-scale topology is specified before statistical estimation
using documented atmospheric, ecological, hydrological, and
geomorphological mechanisms. Physical evidence determines which
directed transitions and parent interactions are admissible and
provides defensible ranges for their lag windows. It does not determine
the numerical magnitude of the corresponding triggering hazards.

For the wildfire and post-fire debris-flow application, a physically
more defensible illustrative topology contains at least four state
components:

\begin{enumerate}
    \item $v_1$: antecedent fuel aridity constructed from drought,
    temperature, and vegetation-moisture indicators;

    \item $v_2$: wildfire occurrence and burn severity;

    \item $v_3$: extreme precipitation during a declared storm window;

    \item $v_4$: post-fire debris-flow occurrence and severity.
\end{enumerate}

The admissible directed edges include
$(v_1,v_2)$, $(v_2,v_4)$, and $(v_3,v_4)$. The interaction set for
the debris-flow node includes $(v_2,v_3)\in\mathcal{I}_4$ because
burned-slope susceptibility and intense precipitation act jointly.
For a specifically post-fire debris-flow target, the at-risk set is
restricted to eligible burned areas during a prespecified recovery
window. Therefore, precipitation outside an eligible burned area
cannot be interpreted as a post-fire debris-flow trigger.

Fuel aridity modifies wildfire susceptibility but does not by itself
represent ignition. Ignition occurrence, wind, fuel continuity, and
suppression conditions enter the wildfire onset model as observed
covariates or, when scientifically justified, as additional nodes.
Similarly, slope, soil properties, burn severity, rainfall intensity,
and antecedent precipitation enter the physical stress function for
debris flow.

Some graph nodes may represent physical drivers without carrying a
direct insured loss. Such nodes remain part of the cascade mechanism
but have zero direct loss capacity. Only loss-bearing nodes are assigned
positive insured-loss parameters.

All admissible parent effects and prespecified interactions are
estimated jointly for each target node. Structural zeros are imposed
only on pathways excluded before estimation. An admissible edge is not
deleted solely because an individual significance test fails. Weakly
identified coefficients are retained with their uncertainty intervals
or subjected to a shrinkage rule specified before test-sample
evaluation.

The resulting structure is described as a mechanistically constrained
directed model. Mechanistic admissibility alone does not establish an
intervention effect. A causal interpretation additionally requires
consistent event definitions, correct temporal ordering, positivity,
adequate adjustment for common causes, and absence of relevant
unmeasured confounding. When these conditions cannot be defended, the
estimated quantities are interpreted as physics-constrained predictive
hazards.

Once the topology, event definitions, and lag windows have been fixed, the heterogeneous physical and insurance data must be placed on compatible spatial and temporal supports without erasing the timing needed to identify each transition. Meteorological variables may be constructed from ERA5-Land or other
documented reanalysis products. Burned area and burn severity may be
obtained from products such as MCD64A1, VNP64A1, or the Monitoring
Trends in Burn Severity archive. Post-fire debris-flow events and
susceptibility variables may be constructed from documented USGS event
inventories and hazard-assessment products. CMIP6 output is used only
to construct future scenario covariates after the historical model has
been estimated.

Every data source must be accompanied by its scientific citation,
product version, download date, quality-assurance flags, coordinate
reference system, native resolution, missing-data treatment, and
spatial and temporal aggregation operator. Products with overlapping
coverage are used for cross-product validation and are not treated as
independent observations of the same event.

The analysis does not force every physical process onto a common weekly
clock. Instead, event construction uses a target-specific observation
scale and an edge-specific lag window. For example, fuel aridity may
be constructed from rolling meteorological and vegetation-moisture
summaries, wildfire onset may be identified on a daily event clock,
and debris-flow triggering precipitation may require storm-level or
subdaily peak-intensity summaries. Weekly aggregation may subsequently
be used for actuarial forecast evaluation, but it must not replace the
higher-frequency variables required to define physical onset.

All spatial products are mapped to a documented common analysis support.
The selected support should be no finer than the least reliable input
needed for the relevant edge. Regridding operators must preserve the
meaning of each variable. Continuous meteorological fields may be
area-averaged or interpolated, whereas burned-area and event indicators
require aggregation rules appropriate for their discrete spatial
support.

The effective estimation window for each edge is determined by the
intersection of the products required to estimate that edge. Longer
meteorological availability cannot be used to claim a longer joint
estimation period when the required fire or debris-flow observations
begin later. Earlier and later satellite products may be linked only
through an explicit cross-calibration procedure estimated within the
training data.

Let $S_i^k(t)\in[0,1]$ denote the observed normalized physical severity
for component $i$ in spatial unit $k$ at its declared observation time.
All normalization rules, physical bounds, transformations, event
thresholds, and missing-value procedures are estimated or specified
using the training period and then frozen before validation and testing.

Define the observed state as
$Z_i^k(t)=\mathbf{1}\{S_i^k(t)\geq u_{i,k}\}$, where $u_{i,k}$ is an
observational event-definition threshold. For target node $j$, the
onset indicator is
$Y_{j,k,t}=\mathbf{1}\{Z_j^k(t)=1,Z_j^k(t-1)=0\}$ and the risk-set
indicator is $R_{j,k,t}=1-Z_j^k(t-1)$. For a post-fire target, the risk
set additionally includes the declared burned-area eligibility and
recovery-window conditions.

For edge $(i,j)$, define the lagged parent contribution as
$Q_{ij,k,t}=\max_{\ell\in\mathcal{L}_{ij}}
Z_i^k(t-\ell)q_{ij}(S_i^k(t-\ell))$. This quantity retains both parent
occurrence and parent severity and corresponds directly to the
severity-modulated hazard contribution in
Equation~\eqref{eq:onset-probability}. For a prespecified parent pair
$(i,\ell)\in\mathcal{I}_j$, the interaction contribution
$Q^{\mathrm{int}}_{i\ell,j,k,t}$ is constructed from the corresponding
joint parent states and physical timing rule.

The observational threshold $u_{j,k}$ defines whether an event is
recorded in the data. It is distinct from the physical stress threshold
$\theta_j$ in Equation~\eqref{eq:severity-response}. The two quantities
must not be estimated or interpreted interchangeably.

\end{document}